\documentclass[a4paper]{article}
\usepackage{anysize}
\marginsize{2.5cm}{2.5cm}{2.5cm}{2.5cm}
\usepackage{authblk}

\usepackage{hyperref,physics}
\usepackage{ifthen}
\usepackage{amsmath,amsthm,amssymb}
\usepackage{verbatim}
\usepackage{array}
\usepackage{mathtools}

\usepackage[export]{adjustbox}

\usepackage{soul}

\usepackage{subcaption}

\usepackage{algorithm2e}
\RestyleAlgo{ruled}

\usepackage{lipsum}
\usepackage[dvipsnames,table]{xcolor}

\usepackage[nowidering]{yhmath}

\newcommand{\N}[0]{\mathbb{N}}

\newcommand{\R}[0]{\mathbb{R}}
\newcommand{\C}[0]{\mathbb{C}}
\newcommand{\Ham}[0]{\operatorname{Ham}}
\newcommand{\Span}[1]{\operatorname{Sp}\!\left(#1\right)}
\renewcommand{\Trace}[1]{\operatorname{tr}\left(#1\right)}

\newcommand{\noinitial}[1]{}

\newcommand{\added}[1]{#1}
\newcommand{\replaced}[2]{#1}
\newcommand{\deleted}[1]{}

\definecolor{SO}{HTML}{e66101}
\definecolor{SOR}{HTML}{fdb863}
\definecolor{ST}{HTML}{5e3c99}
\definecolor{STR}{HTML}{b2abd2}

\usepackage{tikz,pgfplots}
\pgfplotsset{compat = newest}
\usetikzlibrary{pgfplots.groupplots}
\usetikzlibrary{graphs, graphs.standard}

\usetikzlibrary{quantikz}

\usetikzlibrary{fadings,shapes.arrows,shadows,intersections}

\tikzfading[name=arrowfading, top color=transparent!0, bottom color=transparent!95]
\tikzset{arrowfill/.style={top color=blue!20, bottom color=blue, general shadow={fill=black, shadow yshift=-0.8ex, path fading=arrowfading}}}
\tikzset{arrowstyle/.style={draw=blue,arrowfill, single arrow,minimum height=#1, single arrow,
single arrow head extend=.4cm,}}

\usepackage{booktabs}

\usepackage{simplewick}     

\usepackage[
backend=biber,
style=numeric,
url=false,
doi=false
]{biblatex}

\newtheorem{theorem}{Theorem}

\newtheorem{remark}{Remark}
\newtheorem{corollary}{Corollary}

\newtheorem{definition}{Definition}

\title{Constraint preserving mixers for QAOA}
\date{\today}

\author[$\dagger$]{Franz G. Fuchs}
\author[$\dagger$]{Kjetil Olsen Lye}
\author[$\dagger$]{Halvor Møll Nilsen}
\author[$\dagger$]{Alexander J. Stasik}
\author[$\dagger$]{Giorgio Sartor}
\affil[$\dagger$]{SINTEF AS, Department of Mathematics and Cybernetics, Oslo, Norway}

\begin{document}
\maketitle

\abstract{
    The quantum approximate optimization algorithm/quantum alternating operator ansatz (QAOA) is a heuristic 
    to find approximate solutions of combinatorial optimization problems.
    Most literature is limited to quadratic problems without constraints.
    However, many practically relevant optimization problems do have (hard) constraints that need to be fulfilled.
    In this article\added{,} we present a framework for constructing mixing operators that restrict the evolution to a subspace of the full Hilbert space given by these constraints;
    We generalize the ``XY"-mixer designed to preserve the subspace of ``one-hot" states to the general case of subspaces given by a number of computational basis states.
    We expose the underlying mathematical structure which 
    reveals more of how mixers work and how one can minimize their cost in terms of number of \replaced{CX}{CNOT} gates, particularly when Trotterization is taken into account.
    Our analysis also leads to valid Trotterizations for \added{an }``XY"-mixer with fewer \replaced{CX}{CNOT} gates \replaced{than}{as} is known to date.
    In view of practical implementations\added{,} we also describe algorithms for efficient decomposition into basis gates.
    Several examples of more general cases are presented and analyzed.
}

\section{Introduction}
The quantum approximate optimization algorithm (QAOA)~\cite{farhi2014quantum}, and its generalization, the quantum alternating operator ansatz (also abbreviated as QAOA)~\cite{hadfield2019quantum} is a meta-heuristic for solving combinatorial optimization problems that can utilize gate based quantum computers and possibly outperform purely classical heuristic algorithms.
Typical examples that can be tackled are 
quadratic (binary) optimization problems of the form
\begin{equation}
    x^* = \underset{x\in \{0,1\}^n,\ g(x) = 0}{\text{arg min}} \ f(x), \quad f(x) = x^T Q_f x + c_f, \quad  g(x) = x^T Q_g x + c_g
    \label{eq:QCBO}
\end{equation}
where $Q_f,Q_g\in\R^{n\times n}$ are symmetric $n\times n$ matrices.
For binary variables $x\in\{0,1\}$, any linear part can be absorbed into the diagonal of $Q_f$ and $Q_g$.
In this article we focus on the case where the constraint is given by a feasible subspace as defined in the following\replaced{:}{.}
\begin{definition}[Constraints given by indexed computational basis states]
\label{def:B}
Let $\mathcal{H} = (\C^2)^{\otimes n}$
the Hilbert space for $n$ qubits, which is spanned by all computational basis states $\ket*{z_j}$, i.e., 
$\mathcal{H}=\text{\normalfont span}\{\ket*{z_j}, 1\leq j\leq 2^n, z_j\in\{0,1\}^n\}$.
Let
\begin{equation}
    B=\left\{\ket*{z_j}, \ \ j\in J, \ \ z_j\in\{0,1\}^n \right\},
\end{equation}
the subset of all computational basis states defined by an index set J.
This corresponds to
\begin{equation}
    g(x)=\prod_{j\in J} \replaced{\sum_{i=1}^n}{\sum_i} \left(x_i-(z_j)_i\right)^2,
\end{equation}
which is a quadratic constraint.
\end{definition}
There is a well-established connection of quadratic (binary) optimization problems to Ising models, see e.g.~\cite{lucas2014ising}, that allows \added{one} to directly translate these problems to the QAOA.
The general form of QAOA is given by
\begin{equation}
    \ket{\gamma,\beta} = U_M(\beta_p) U_P(\gamma_p)\cdots U_M(\beta_1) U_P(\gamma_1) \ket{\phi_0},
\end{equation}
where one alternates the application of phase separating and mixing operator $p$ times.
Here, $U_P(\gamma)$ is a phase separating operator that depends on the objective function $f$.
As defined in \cite{hadfield2019quantum} the \ul{requirements for the mixing operator} $U_M(\beta)$ are as follows
\begin{itemize}
    \item $U_M$ does not commute with $U_P$, i.e., $[U_M(\beta),U_P(\gamma)]\neq 0$, for almost all $\gamma,\beta\in\R$,
    \item $U_M$ \textit{preserves the feasible subspace} as given in Definition~\ref{def:B}, i.e., $\Span{B}$ is an invariant subspace of $U_M$,
    \begin{equation}
    \label{eq:mixer_preserve}
        U_M(\beta)\ket{v} \in \Span{B}, \quad \forall \ket{v} \in \Span{B}, \forall \beta\in\R,
    \end{equation}
    \item $U_M$ provides transitions between all pairs of feasible states, i.e., for each pair $x,y$
    $\exists \beta^*\in\R$ 
    and $\exists r \in \N\cap\{0\}$
    , such that
    \begin{equation}
    \label{eq:mixer_transition}
        |\bra{x}
        \underbrace{U_M(\beta^*)\cdots U_M(\beta^*)}_{r\text{ times}}
        \ket{y}|
            >0, \forall \text{ comp. basis states } \ket{x}, \ket{y} \in B.
            \end{equation}
\end{itemize}
If both $U_M$ and $U_P$ correspond to \replaced{the time}{th etime} evolution under som\added{e} Hamiltonians $H_M, H_P$,
\added{i.e., $U_M=e^{-i\beta H_M}$ and $U_P=e^{-i\gamma H_P}$,}
the approach can be termed ``Hamiltonian-based QAOA" (H-QAOA). If the Hamiltonians $H_M, H_P$ are the sum of (polynomially many) local terms it represents a sub-class termed ``local Hamiltonian-based QAOA" (LH-QAOA).

\added{
In practice, it is not possible to implement $U_M$ or $U_P$ directly.
It is necessary to decompose the evolution into smaller pieces, which means that instead of applying $e^{-it(H_1+H_2)}$ one can only apply $e^{-i t H_1}$ and $e^{-i t H_2}$.
This process is typically referred to as \emph{``Trotterization"}.
%
%
As an example, the simplest Suzuki-Trotter decomposition, or the exponential product formula~\cite{hatano2005finding,trotter1959product} is given by
\begin{equation}
    e^{x(H_1+H_2)} = e^{x H_1} e^{x H_2} + \mathcal{O}(x^2)
\end{equation}
where $x$ is a parameter and $H_1, H_2$ are two operators with some commutation relation $[H_1,H_2]\neq 0$.
Higher order formulas can be found for instance in~\cite{hatano2005finding}.
}

\added{
Practical algorithms need to be defined using a few operators from a universal gate set, e.g., $\{U_3,CX\}$, where
\begin{equation}
U_3(\theta,\phi,\lambda) =
\begin{pmatrix}
\cos(\theta/2) & -e^{i\lambda} \sin(\theta/2) \\
e^{i\phi} \sin(\theta/2) & e^{i(\phi+\lambda)}\cos(\theta/2) \\
\end{pmatrix}, \quad CX =
\begin{pmatrix}
1 & 0 & 0 & 0 \\
0 & 1 & 0 & 0 \\
0 & 0 & 0 & 1 \\
0 & 0 & 1 & 0 \\
\end{pmatrix}.
\end{equation}
A good (and simple) indicator for the complexity of a quantum algorithm is given by the number of required $CX$ gates.
Overall, the most efficient algorithm is the one that provides the best accuracy in a given time~\cite{kronsjo1987algorithms}.
}

\begin{remark}[Repeated mixers]
If $U_M$ is the exponential of a Hermitian matrix, the parameter $r$ \added{in Equation~\eqref{eq:mixer_transition}} does not matter as it can be absorbed as a re-scaling of $\beta$.
However, if $U_M$ is Trotterized\deleted{, i.e., the product of exponentials of non-commuting Hermitian matrices} this can lead to missing transitions.
In this case $r>1$ can again provide these transitions. It is therefore suggested \deleted{as useful} in~\cite{hadfield2019quantum} to repeat mixers within one mixing step.
\replaced{For}{Out of} this reason, we will \replaced{consider}{treat} the cost of Trotterized mixers including the necessary repetitions to provide transitions for all feasible states.
\end{remark}


\section{Related work}
\added{
The QAOA was introduced by \cite{farhi2014quantum} where it was applied to the Max-Cut problem.
The authors in~\cite{guerreschi2019qaoa} compared the QAOA to the classical AKMAXSAT solver extrapolate from small instances to large instances and estimate that a quantum speed-up can be obtained with (several) hundreds of qubits.
%
A general overview of variational quantum algorithms, including challenges and how to overcome them, is provided in \cite{cerezo2021variational,Moll2018}.
Key challenges are that it is in general hard to find good parameters.
It has been shown that the training landscapes are in general NP-hard~\cite{bittel_training_2021}.
Another obstacle are so-called barren plateaus, i.e. regions in the training landscape where the loss function is
effectively constant~\cite{Moll2018}.
This phenomenon can be caused by random initializations, noise, and over-expressablity of the ansatz\cite{wang2021noise,zhang_fundamental_2022}
%
}

\added{
Since its inception, several extensions/variants of the QAOA have been proposed. 
\textit{ADAPT-QAOA}~\cite{zhu2020adaptive} is an iterative, problem-tailored version of QAOA that can adapt to specific hardware constraints.
A non-local version, referred to as \textit{R-QAOA}~\cite{bravyi2019obstacles} recursively removes variables from the Hamiltonian until the remaining instance is small enough to be solved classically.
Numerical evidence shows that this procedure significantly outperforms standard QAOA for frustrated Ising models on random 3-regular graphs for the Max-Cut problem.
\textit{WS-QAOA}~\cite{egger2020warm} takes into account solutions of classical algorithms to a warm-starting QAOA.
Numerical evidence shows an advantage at low depth, in the form of a systematic increase in the size of the obtained cut for fully connected graphs with random weights.
}

There are two principal ways to take constraints into account when solving Equation~\eqref{eq:QCBO} with the QAOA.
The standard, simple approach is to penalize unsatisfied constraints in the objective function with the help of a so called Lagrange multiplier $\lambda$, leading to
\begin{equation}
    x^* = \underset{x\in \{0,1\}^n}{\text{arg min}} \ \left( f(x) + \lambda g(x) \right).
\end{equation}
This approach is popular, since it is straightforward to define a phase separating Hamiltonian for $f(x) + \lambda g(x)$.
Some applications include the tail-assignment problem~\cite{PhysRevApplied.14.034009}, the Max-k-cut problem~\cite{fuchs2021efficient}, graph coloring problems, and the traveling sales person problem~\cite{hadfield2017quantum}.
A downside of this \deleted{is} approach is that infeasible solutions are also possible outcomes, especially for approximate solvers like QAOA. This also makes the search space much bigger and the entire approach less efficient.
In addition, the quality of the results turns out to be very sensitive to the chosen value of the hyper parameter $\lambda$. On one hand, $\lambda$ should be chosen large enough such that the lowest eigenstates of $H_P$ correspond to feasible solution\added{s}. On the other hand, too large values of $\lambda$ mean that the resulting optimization landscape in the $\gamma$ has very high frequencies, which makes the problem hard to solve in practice. In general, it can be very challenging to find (the problem dependent) value for $\lambda$ that \replaced{best}{better} balances the trade off between optimality and feasibility in the objective function~\cite{Wang2020}.

For QAOA, a second approach is to define mixers that have zero probability to go from a feasible state to an infeasible one, making the hyper parameter $\lambda$ of the previous approach unncessary.
However, it is generally more challenging to devise mixers that take into account constraints.
The most prominent example in \added{the} literature is the $XY$-mixer~\cite{hadfield2017quantum,hadfield2019quantum,Wang2020} which constrains evolution to states with non-zero overlap with ``one-hot" states. One-hot states are computational basis states with exactly one entry equal to one. For instance $\ket{0001}$ and $\ket{010000}$ are one-hot states, while $\ket{00}$ and $\ket{110}$ are not.
The name $XY$ mixer comes from the related $XY$-Hamiltonian~\cite{lieb1961two}.
The mixers derived in \added{the} literature follow the intuition of physicists to use ``hopping" terms.
A performance analysis of the XY-mixer applied to the maximum k-vertex cover shows a heavy dependence on the initial states as well as the chosen Trotterization~\cite{cook2020quantum}.

QAOA can be viewed as a discretized version of quantum annealing. In quantum annealing enforcing constraints via penalty terms is particularly ``harmful" since they often require all-to-all connectivity of the qubits~\cite{hen2016driver}.
The authors in~\cite{hen2016quantum} therefore introduce driver Hamiltonians that commute with the constraints of the problem.
This bears similarities with and actually inspired the approaches in~\cite{hadfield2017quantum,hadfield2019quantum}.




The main contributions of this article are:
\begin{itemize}
    \setlength\itemsep{0\baselineskip}
    \item A general framework to construct mixers restricted to a set of computational basis states, see Section~\ref{sec:condmixers}.
    \item An analysis of the underlying mathematical structure, which is largely independent of the actual states, see Section~\ref{sec:transmatrix}.
    \item Efficient algorithms for decomposition into basis gates, see Section~\ref{sec:decomp} and~\ref{sec:Trotterizations}.
    \item Valid Trotterizations, which is not completely understood in \added{the} literature, see Section~\ref{sec:Trotterizations}.
    \item \added{We prove that it is always possible to realize a valid Trotterization, see Theorem~\ref{theorem:P_Tij_commutes}.}
    \item Improved efficiency of Trotterized mixers for ``one-hot" states in Section~\ref{sec:XYmixer}.
    \item Discussion of the general case, exemplified in Section~\ref{sec:generalcases}.
\end{itemize}
We start by describing the general framework.

\section{Construction of constraint preserving mixers}

In the following we will derive a general framework for mixers that are restricted to a subspace, given by certain basis states.
For example, one may want to construct a mixer for five qubits that is restricted to the subspace $\Span{\ket{01001}, \ket{11001}, \ket{11110}}$ of \replaced{$\C^{2^5}$}{$\C^{2^n}$}, where $\Span{B}$ denotes the linear span of $B$.
In this section we will describe the conditions for a Hamiltonian-based QAOA mixer to preserve the feasible subspace, and for providing transitions between all pairs of feasible states.
We also provide efficient algorithms to decompose these mixers into basis gates.


\subsection{Conditions on the mixer Hamiltonian}\label{sec:condmixers}

\begin{theorem}[Mixer Hamiltonians for subspaces]\label{theorem:MixerHamiltonian}
Given a feasible subspace $B$ as in Definition~\ref{def:B}
and a real-valued \emph{transition} matrix $T\in\R^{|J|\times|J|}$.
Then, for the mixer constructed via
\begin{equation}
    U_M(\beta) = e^{-i\beta H_M}, \quad \text{ where } H_M = \sum_{j,k\in J} (T)_{j,k}\ket*{x_j}\bra*{x_k},
    \label{eq:Hmdefinition}
\end{equation}
the following statements hold.
\begin{itemize}
    \item If $T$ is symmetric, the mixer is well defined and preserves the feasible subspace, i.e. condition~\eqref{eq:mixer_preserve} is fulfilled.
\item
If $T$ is symmetric and for all $1\leq j,k \leq |J|$ 
there exist\added{s} an $r\in\N\cup\{0\}$ (possibly depending on the pair) such that
\begin{equation}
    (T^r)_{j,k} \neq 0,
    \label{eq:Hmfeasiblecond}
\end{equation}
then $U_M$ provides transitions between all pairs of feasible states, i.e. condition~\eqref{eq:mixer_transition} is fulfilled.
\end{itemize}
\end{theorem}

\begin{figure}
    \centering
    \begin{subfigure}[t]{0.48\textwidth}
    \centering
        \begin{tikzpicture}[scale=.225]
            \node[] at (-5, 4)   (a) {$\C^{2^n}$};
            \node[] at ( 5, 4)   (d) {$\C^{2^n}$};
            \node[] at (-5,-4)   (b) {$\C^{|J|}$};
            \node[] at ( 5,-4)   (c) {$\C^{|J|}$};
            
            \draw[thick,-{Latex[length=4mm]}] (a) -- node[above]{$H_{M,B}$} (d);
            \draw[thick,-{Latex[length=4mm]}] (a) -- node[left]{$E^T$} (b);
            \draw[thick,-{Latex[length=4mm]}] (b) -- node[below]{$T$} (c);
            \draw[thick,-{Latex[length=4mm]}] (c) -- node[right]{$E$} (d);
        \end{tikzpicture}
        \caption{The action of the Hamiltonian $H_{M,B}$ can be understood as the transition matrix $T$ acting upon the feasible basis states of $B =\left\{\ket*{x_j}, \ \ j\in J \right\}$.
    \deleted{Here, $E = \{ \ket*{x_j} \}_{j\in J}$.}
    }
    \end{subfigure}
    \hfill
    \begin{subfigure}[t]{0.48\textwidth}
        \centering
        \begin{tikzpicture}[scale=.3]

        \begin{scope}[fill opacity=0.5,text opacity=1]
        
        \def\shifty{0}
        \def\shift{12}
        
        \draw[draw = black, fill=gray] (0,0) circle (5);
        \draw[fill=green, draw = black,name path=circle 1] (0,2) circle (2);

        \node[] at (0,6) (A) {$\C^{2^n}$};
        \node[] at (0,2+\shifty) (B2) {$\Span{B}$};
        
        \draw[] (0,0) node[circle,inner sep=1pt,fill,fill opacity=1,label=below:$0$](O){};
        
        \draw[draw = black, fill=gray] (0+\shift,0) circle (5);
        \draw[fill=green, draw = black,name path=circle 1] (0+\shift,2+\shifty) circle (2);
        
        \node[] at (0+\shift,6) (A2) {$\C^{2^n}$};
        \node[] at (0+\shift,2+\shifty) (B2) {$\Span{B}$};
        \draw[] (0+\shift,0) node[circle,inner sep=1pt,fill,fill opacity=1,label=below:$0$](O){};
        
        \draw[gray,very thick,-{Latex[length=4mm]},opacity=.75] (2,2) -- (-2+\shift,2+\shifty);
        \draw[green,very thick,-{Latex[length=4mm]},opacity=.75] (2,2) -- (-2+\shift,2+\shifty);
        
        \draw[gray,very thick,-{Latex[length=4mm]},opacity=1] (2,0) -- (0+\shift,0+\shifty);
        
                
        \node[] at (0.5*\shift,3+\shifty) (H) {$H_{M,B}$};

        \end{scope}
        \end{tikzpicture}
        \caption{The Hamiltonian $H_{M,B}:\C^{2^n} \rightarrow \C^{2^n}$ maps everything outside $\Span{B}$ to $0$ and $\Span{B}$ onto itself.
        }
    \end{subfigure}

    \caption{
    Illustration of properties of Hamiltonians constructed with Theorem~\ref{theorem:MixerHamiltonian}.
    }
    \label{fig:H_M_comm}
\end{figure}
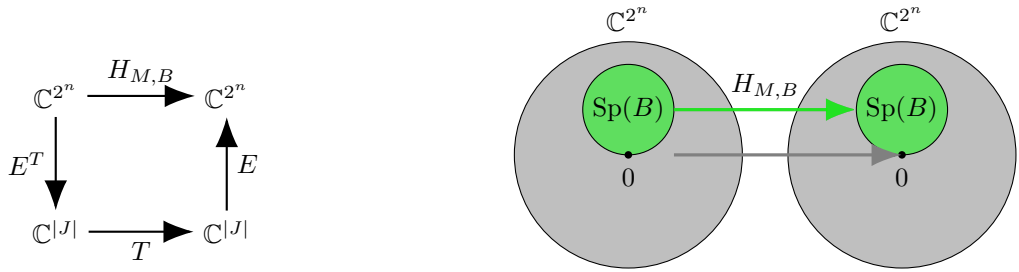

\begin{proof}

\noindent \ul{Well definedness.}

\noindent
Almost trivially $H_M$ is Hermitian if $T$ is symmetric,
\begin{equation}
    H_M^\dagger = \sum_{j,k\in J} (T)_{j,k}\ket*{x_k}\bra*{x_j} = \sum_{j,k\in J} (T)_{k,j}\ket*{x_k}\bra*{x_j} = H_M.
\end{equation}
%
%
Since $H_M$ is a Hermitian (and therefore normal) matrix 
there exists a diagonal matrix $D$, with the entries of the diagonal as the (real valued) eigenvalues of $H_M$, and a matrix $U$, with columns given by the corresponding orthonormal eigenvectors.
%
The mixer is therefore well defined through the convergent series
\begin{equation}\label{eq:eH}
    e^{-i t H_M} = 
    \sum_{m=0}^\infty \frac{(-it)^m H_M^m}{m!} = U e^{-i t D} U^\dagger.
\end{equation}

\noindent\ul{Reformulations.}

\noindent
We can rewrite $H_M$ in the following way
\begin{equation}
        H_M: \ \ \ket{y} \mapsto \sum_{j,k\in J} (T)_{j,k}\bra*{x_j}\ket{y} \ket*{x_k} = E T E^T \ket{y}, \ \ \ket{y}\in\C^{2^n},
    \label{eq:H_ETET}
\end{equation}
where the columns of the matrix $E \in \R^{2^n\times |J|}$ consist\deleted{s} of the feasible computational basis states, i.e., \replaced{$E = [x_j ]_{j\in J}$}{$E = \{ \ket{x_j} \}_{j\in J}$}, see Figure~\ref{fig:H_M_comm} for an illustration.
Using that $E^TE=I \in\R^{|J| \times |J|}$ is the identity matrix, we have that
\begin{equation}\label{eq:Hpow}
    H_M^m = ET^m E^T = \sum_{j,k\in J} (T^m)_{j,k}\ket*{x_j}\bra{x_k}, \ m \in \N,
\end{equation}
and Equation~\eqref{eq:eH} can be written as
\begin{equation}\label{eq:eitH}
    e^{-i t H_M} = 
    E\left(\sum_{m=0}^\infty \frac{(-it)^m T^m}{m!}\right)E^T.
\end{equation}

\noindent\ul{Preservation of the feasible subspace.}

\noindent
Let $\ket{v}\in \Span{B}$. Using Equation~\eqref{eq:Hpow} we know that
$$H_M^m\ket{v} = \sum_{j,k\in J} (T^m)_{j,k}\ket*{x_j}\bra{x_k}\ket{v} = \sum_{j\in J} c_j \ket{x_j} \in \Span{B},$$ with coefficients $c_j\in\C$.
Therefore, also $e^{-i t H_M} \ket{v} \in \Span{B}, t \in \R$, since it is a sum of these terms.

\noindent\ul{Transition between all pairs of feasible states.}

\noindent
For any pair of feasible computational basis states $\ket{x_{j^*}}, \ket*{x_{k^*}} \in B$ we have that
\begin{equation}\label{eq:overlap}
\begin{split}
f(t)&=
    \bra{x_{j^*}} U_M(t) \ket*{x_{k^*}}
    =
    \bra{x_{j^*}} \sum_{m=0}^\infty \left(\frac{(-it)^m}{m!} \sum_{j,k\in I} (T^m)_{j,k}\ket*{x_j}\bra{x_k}\right) \ket*{x_{k^*}}\\
    &=
    \sum_{m=0}^\infty \frac{(-it)^m}{m!} (T^m)_{j^*,k^*}
\end{split}
\end{equation}
It is enough to show that $f(t)$ is not the zero function.
Since $f(t):\R\rightarrow\C$ is an analytic function it has a unique extension to $\C$.
Assume that $f$ is \deleted{were} indeed the zero function on $\R$, then the extension to $\C$ would also be the zero function and all coefficients of its Taylor series \replaced{would be}{were} zero.
However, we assumed the existence of an $r\in\N\cup\{0\}$ such that $|(T^r)_{j^*,k^*}|>0$, and hence there exists a non-zero coefficient, which is a contradiction to $f$ being the zero function.

\end{proof}


A natural question is how the statements in Theorem~\ref{theorem:MixerHamiltonian} depend on the particular ordering of the elements of $B$.

\begin{corollary}[Independence of the ordering of $B$.]
Statements in Theorem~\ref{theorem:MixerHamiltonian} that hold for a particular ordering of computational basis states for a given $B$, hold also for any permutation $\pi:\{1,\cdots,|J|\} \rightarrow \{1,\cdots,|J|\}$, i.e., they are independent of the ordering of elements.
For each ordering, the transition matrix $T$ changes according to $T_\pi = P_\pi^T T P_\pi$, where $P_\pi$ is the permutation matrix associated with $\pi$.
\label{corollary:independenceofordering}
\end{corollary}
\begin{proof}
We start by pointing out that the inverse matrix of $P_\pi$ exists and can be written as $P_\pi^{-1} = P_{\pi^{-1}} = P_\pi^T$.

\noindent\ul{The resulting matrix $H_M$ is unchanged.}
Following the derivation in Equation~\eqref{eq:H_ETET}, we have that
$
    H_{M_\pi} = E_\pi T_\pi  E_\pi^T,
$
where the columns of the matrix $E \in \R^{2^n\times |J|}$ consists of the \textit{permuted} feasible computational basis states, i.e., $E_\pi=\{x_{\pi(j)}\}_{j\in{J}}$.
Inserting $T=P_\pi^T T P_\pi$ we have indeed
$
    H_{M_\pi} = E_\pi T_\pi  E_\pi^T = (E_\pi P_\pi^T ) T (P_\pi E_\pi^T) = E T E^T = H_M
$.


\noindent\ul{$T_\pi$ is symmetric if $T$ is.}
Assuming that $T^T=T$ we have that also
$$
    (T_\pi)^T = (P_\pi^T T P_\pi)^T = P_\pi^T T^T P_\pi = P_\pi^T T P_\pi = T_\pi.
$$

\noindent\underline{If the condition in Equation~\eqref{eq:Hmfeasiblecond} holds for $T$ than it also holds for $T_\pi$.}
Using $T_\pi^r = P_\pi^T T^r P_\pi$ we can show that Equation~\eqref{eq:Hmfeasiblecond} holds for the  permuted index pair $(\pi(j), \pi(k))$ for $T_\pi$ if it holds for $(j,k)$ for $T$.
\end{proof}

    In the following, if nothing else is remarked, computational basis states are ordered with respect to increasing integer value, e.g., $\ket{001}, \ket{010}, \ket{111}$.

    Apart form special cases, there is a lot of freedom to choose the transition matrix T that fulfills the conditions of Theorem~\ref{theorem:MixerHamiltonian}. The entries of T will heavily influence the circuit complexity, which will be investigated in Section~\ref{sec:decomp}.
In addition, we have the following property which adds additional flexibility to develop efficient mixers.
\begin{corollary}[Properties of mixers]\label{corollary:properties}
    For a given feasible subspace $\Span{B}$ let $U_{M,B}$ the mixer given by Theorem~\ref{theorem:MixerHamiltonian}.
    For any subspace $\Span{C}$ with $\Span{B}\cap \Span{C}=\{0\}$ or equivalently $B\cap C=\varnothing$, also $U_M=U_{M,B}U_{M,C}$ is a valid mixer for $B$ satisfying the conditions of Equations~\eqref{eq:mixer_preserve} and \eqref{eq:mixer_transition}, see also Figure~\ref{fig:addH}.
\end{corollary}
\begin{proof}
Any $\ket{v} \in B$ is in the null space of $H_{M,C}$, i.e., $H_{M,C} \ket{v} = 0$ and hence $U_{M,C} \ket{v} = I$.
Therefore, $U_{M,B}U_{M,C}\ket{v} = U_{M,B}\ket{v}\in B$, and $U_{M,C}U_{M,B}\ket{v} = U_{M,C}\ket{w} = \ket{w}$ with $\ket{w} \in B$ which means the feasible subspace is preserved.
Condition~\eqref{eq:mixer_transition} follows similarly form the fact that $U_{M,C} \ket{v} = I$ for any $\ket{v} \in B$.
\end{proof}
Corollary~\ref{corollary:properties} naturally holds as well for any linear combination of mixers, i.e., $H_{M,B} + \sum_i a_i H_{M,C_i}$ is a mixer for the feasible subspace $\Span{B}$ as long as $\Span{C_i}\cap\Span{B}=\{0\}, \forall i$.
At first, it might sound counter intuitive that adding more terms to the mixer results in more efficient decomposition into basis gates. However, as we will see in \added{Section}~\ref{sec:constrainedmixers}, it can lead to cancellations due to symmetry considerations.

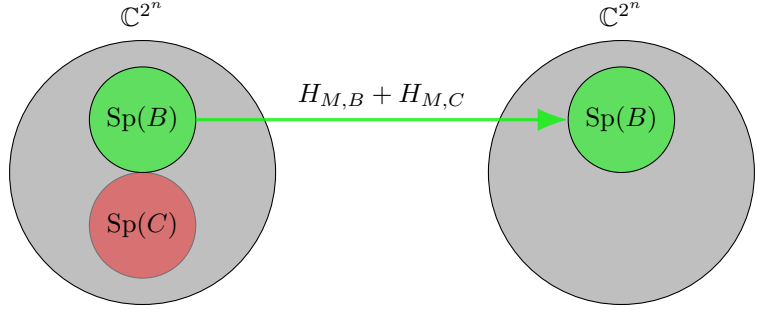
\begin{figure}
    \centering
        \begin{tikzpicture}[scale=.35]

        \begin{scope}[fill opacity=0.5]
        \draw[draw = black,fill=gray] (0,0) circle (5);
        
        \def\shifty{0}
        \def\shift{18}
        
        \draw[fill=green, draw = black,name path=circle 1] (0,2+\shifty) circle (2);
        \draw[fill=red, draw = black,name path=circle 2,opacity=.4] (0,-2+\shifty) circle (2);

        \node[text opacity=1] at (0,6) (A) {$\C^{2^n}$};
        \node[text opacity=1] at (0,-2+\shifty) (B) {$\Span{C}$};
        \node[text opacity=1] at (0,2+\shifty) (C) {$\Span{B}$};
        
        \draw[draw = black,fill=gray] (0+\shift,0) circle (5);
        \draw[fill=green, draw = black,name path=circle 1] (0+\shift,2+\shifty) circle (2);
        
        \node[text opacity=1] at (0+\shift,6) (A2) {$\C^{2^n}$};
        \node[text opacity=1] at (0+\shift,2+\shifty) (B2) {$\Span{B}$};
        
        
        \draw[gray,very thick,-{Latex[length=4mm]},opacity=.25] (2,2) --  (-2+\shift,2+\shifty);
        \draw[green,very thick,-{Latex[length=4mm]},opacity=.75] (2,2) --  node[above,text opacity=1,text=black]{$H_{M,B}+H_{M,C}$}  (-2+\shift,2+\shifty);

        \end{scope}
        \end{tikzpicture}
        \caption{Corollary~\ref{corollary:properties} shows that adding a mixer with support outside $\Span{B}$ is also a valid mixer for $B$.
        }
    \label{fig:addH}
\end{figure}

Next, we describe the structure of the eigensystem of $U_M$.

\begin{corollary}[Eigensystem of mixers]
Given the setting in Theorem~\ref{theorem:MixerHamiltonian} with a symmetric transition matrix $T$.
Let $(\lambda, v)$ an eigenpair of $T$, then $(\lambda, E v)$ is an eigenpair of $H_M$ and $(e^{-it\lambda},  Ev)$ an eigenpair of $U_M$,
where $E = \{ \ket{x_j} \}_{j\in J}$ as defined in Equation~\eqref{eq:H_ETET}.
\label{corollary:feasibleinitialstates}
\end{corollary}
\begin{proof}
Let $(\lambda, v)$ an eigenpair of $T$. Then,
$H_M E v = E T E^T E v = E T v = \lambda E v$,
so $(\lambda, Ev)$ is an eigenpair of $H_M$.
The connection between $H_M$ and $U_M$ is general knowledge from linear algebra.
\end{proof}
An example\deleted{s} illustrating Corollary~\ref{corollary:feasibleinitialstates} is provided by the transition matrix $T\in\R^{4\times4}$ with zero diagonal and all other entries equal to one. A unit eigenvector of T, which fulfills Theorem~\ref{theorem:MixerHamiltonian}, is $v=1/2(1,1,1,1)^T$.
For any $B=\{\ket{z_1}, \ket{z_2},$ $\ket{z_3}, \ket{z_4}\}$ 
the uniform superpositions of these states is an eigenvector, since 
$$
    \frac{1}{\|v\|_2} Ev = \frac{1}{2}(\ket{z_1}, \ket{z_2}, \ket{z_3}, \ket{z_4})(1,1,1,1)^T =
    \frac{1}{2}(\ket{z_1} + \ket{z_2} + \ket{z_3} + \ket{z_4}).
$$
This result holds irrespective of what the states are and which dimension they have.

\begin{theorem}[Products of mixers for subspaces]\label{theorem:TrotterizedMixerHamiltonian}
Given the same setting as in Theorem~\ref{theorem:MixerHamiltonian}.
For any decomposition of $T$ into a sum of $Q$ symmetric matrices $T_q$, in the following sense
\begin{equation}
    T=\sum_{q=1}^Q T_q, \quad  (T_q)_{i,j} = (T_q)_{j,i} =
    \begin{cases}
        \text{either } &(T)_{ij}, \\
        \text{or } &0,
    \end{cases}
\end{equation}
we construct the mixing operator via
\begin{equation}
    U_M(\beta) = \prod_{\underset{q_n\in \{1,2,\cdots,Q\}}{n=1}}^N e^{-i\beta T_{q_n}}.
\end{equation}
If all entries of $T$ are positive, then $U_M$ provides transitions between all pairs of feasible states, i.e. condition~\eqref{eq:mixer_transition} is fulfilled, if for all $1\leq j,k \leq |J|$ 
there exist $r_m\in\N\cup\{0\}$ (possibly depending on the pair) such that
\begin{equation}
 \Big(\prod_{\underset{q_m\in Q}{m=1}}^M T_{q_m}^{r_m}\Big)_{j,k} \neq 0.
    \label{eq:Hmfeasiblecondsum}
\end{equation}
\end{theorem}
\begin{proof}
Combining Equations~\eqref{eq:Hpow} and~\eqref{eq:eitH} we have \deleted{that}
\begin{equation}
\begin{split}
    \bra{x_{j}} U_M(\beta) \ket*{x_{k}} =& 
    \sum_{j_1=0,j_2=0,\cdots, j_M=0}^\infty 
    \frac{(-it)^{j_1+j_2+\cdots+j_m}
    (T_{q_1}^{j_1}T_{q_2}^{j_2}\cdots T_{q_M}^{j_M})_{j,k}}
    {j_1!j_2!\cdots j_m!}
    \\
    =&
    \sum_{j=1}^\infty 
    \frac{(-it)^j}{j!}
    \sum_{j_1,\cdots,j_M \text{ s.t. } \left(\sum_{m=1}^M j_m\right)=j}
    (T_{q_1}^{j_1}T_{q_2}^{j_2}\cdots T_{q_M}^{j_M})_{j,k}
    \added{.}
\end{split}
\end{equation}
Using that $T$ only \added{h}as positive entries and the condition in Equation~\eqref{eq:Hmfeasiblecondsum}\added{,} the same argument as in Theorem~\ref{theorem:MixerHamiltonian} can be used to show that $U_M(\beta)$ is not the zero function and therefore we have transitions between all pairs of feasible states.
\end{proof}
    
As Theorem~\ref{theorem:MixerHamiltonian} 
leaves a lot of freedom for choosing valid transition matrices we will continue by describing important examples for $T$.

\subsection{Transition matrices for mixers}\label{sec:transmatrix}

\begin{figure}
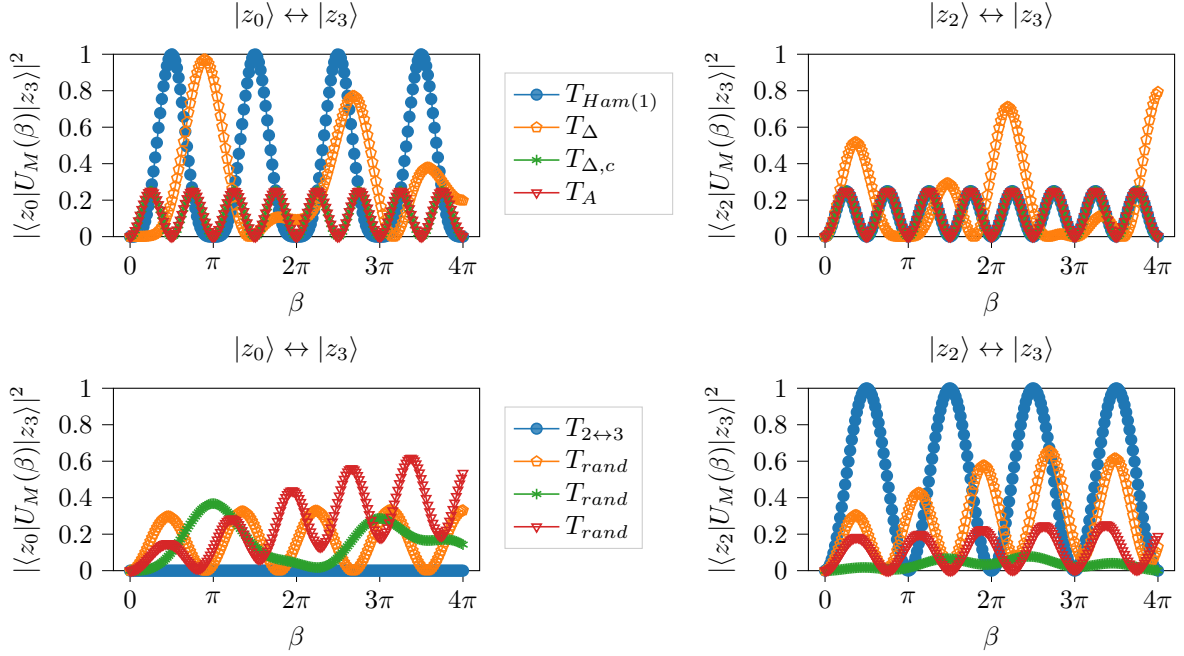

    \centering
        \centering
        \input{snf_00-11.tex}
        \input{snf_10-11.tex}
    \vfill
        \input{orr_00-11.tex}
        \input{orr_10-11.tex}
        
    \caption{Examples of \added{the squared} overlap between two states for the case $|B|=4$. The \added{squared} overlap is independent of what the states in $B=\{\ket{z_0}, \ket{z_1}, \ket{z_2}, \ket{z_3}\}$ are.
    The comparison for different $T$ show that there exists a $\beta$ such that the overlap is nonzero,
    except for $T_{2\leftrightarrow 3}$ which, as expected, does not provide transitions between $\ket{z_0}$ and $\ket{z_3}$.
    }
    \label{fig:overlap}
\end{figure}
Theorem~\ref{theorem:MixerHamiltonian} provides conditions for the construction of mixer Hamiltonians that preserve the feasible subspace and provide transitions between all pairs of feasible computational basis states, namely
\begin{enumerate}
    \item $T\in\R^{|J|\times|J|}$ is symmetric, and
    \item for all $1\leq j,k \leq |J|$ there exist\added{s} an $r_{j,k}\in\N\cup\{0\}$ such that $(T^r)_{j,k} \neq 0$.
\end{enumerate}
Remarkably, these conditions depend only on the dimension of the feasible subspace $|J|=dim(\Span{B})=|B|$, \replaced{and}{but} are \ul{independent of the specific states} that $B$ consists of.
In addition, Corollary~\ref{corollary:independenceofordering} shows that these conditions are robust with respect to reordering of rows if in addition columns are reordered in the same way.
Moreover, Equation~\eqref{eq:overlap} shows \replaced{also that}{that also} the overlap between computational basis states $\ket{x_j},\ket{x_k}\in B$ is \ul{independent of the specific states} that $B$ consists of and only depends on T, since the right hand side of the expression
\begin{equation}
  \bra{x_{j}} U_M(t) \ket*{x_{k}} =
    \sum_{m=0}^\infty \frac{(-it)^m}{m!} (T^m)_{j,k},
\end{equation}
is independent of \added{the} elements in $B$.
This allows \added{us} to describe and analyze valid transition matrices by only knowing the number of feasible states, i.e., $|B|$. What these specific states are is irrelevant, unless one wants to look at what an optimal mixer is, which we will come back to in Section~\ref{sec:optimality}.
Figure~\ref{fig:overlap} provides a comparison of some mixers described in the following with respect to the overlap between different states.

In the following, we denote the matrix for pairs of indices whose binary representation have Hamming distance equal to $d$ as 
\begin{equation}
    T_{\Ham(d)}, \quad \text{ with } ( T_{\Ham(d)})_{i,j}=\begin{cases}
    1, & \  \text{if } d_\text{Hamming}($bin$(i)$, bin$(j))=d,\\
    0, & \ \text{else},
    \end{cases}
    \label{eq:THamming}
\end{equation}
Examples of the structure of $T_{\Ham(d)}$ can be found in Figure~\ref{fig:Tn48}.
\begin{figure}
    \centering
    \begin{subfigure}{0.24\textwidth}
        \centering
        \includegraphics[width=\textwidth]{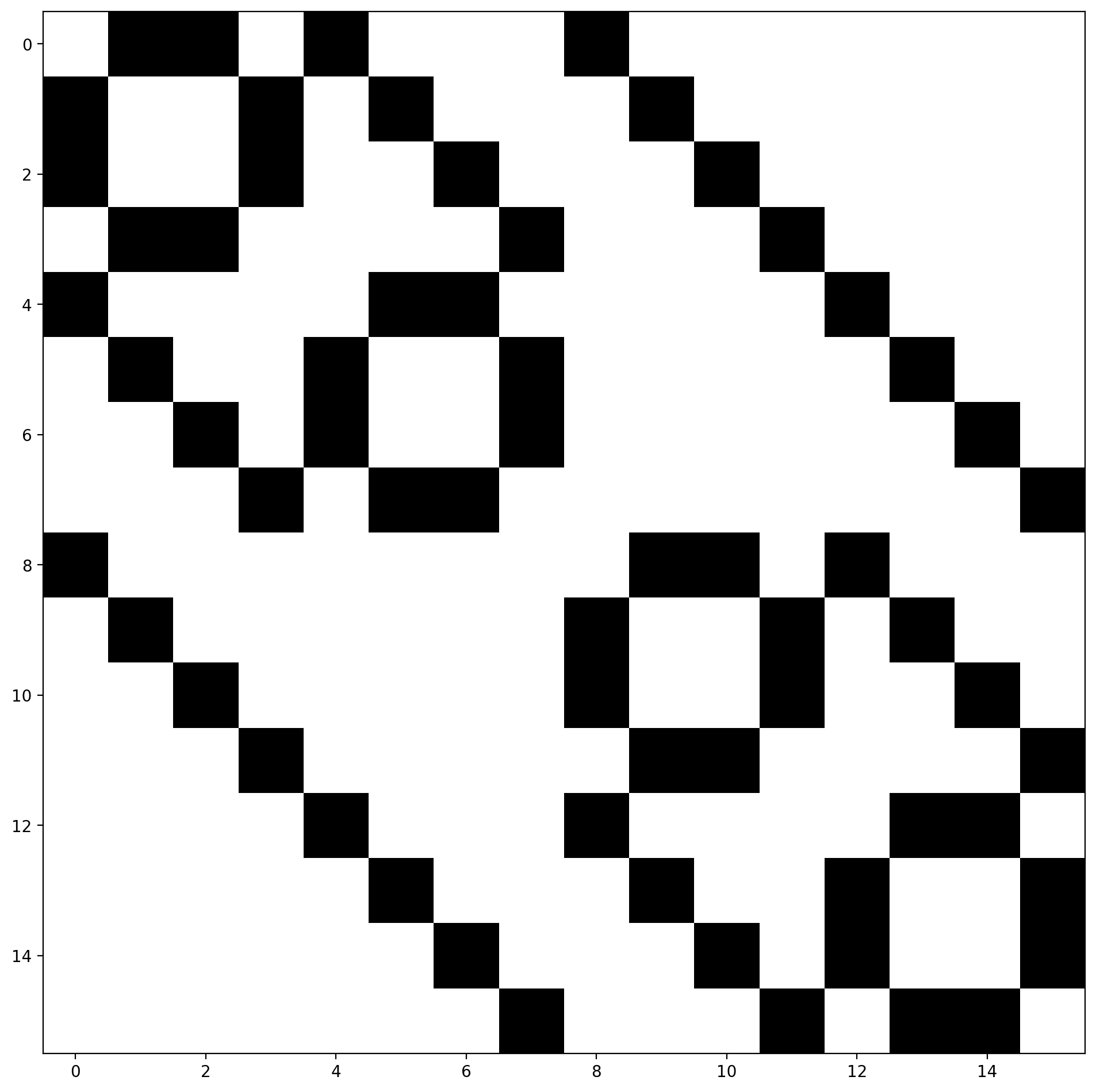}
        \caption{$T_{\Ham(1)}$, $|J|=2^4$.}
    \end{subfigure}
    \hfill
    \begin{subfigure}{0.24\textwidth}
        \includegraphics[width=\textwidth]{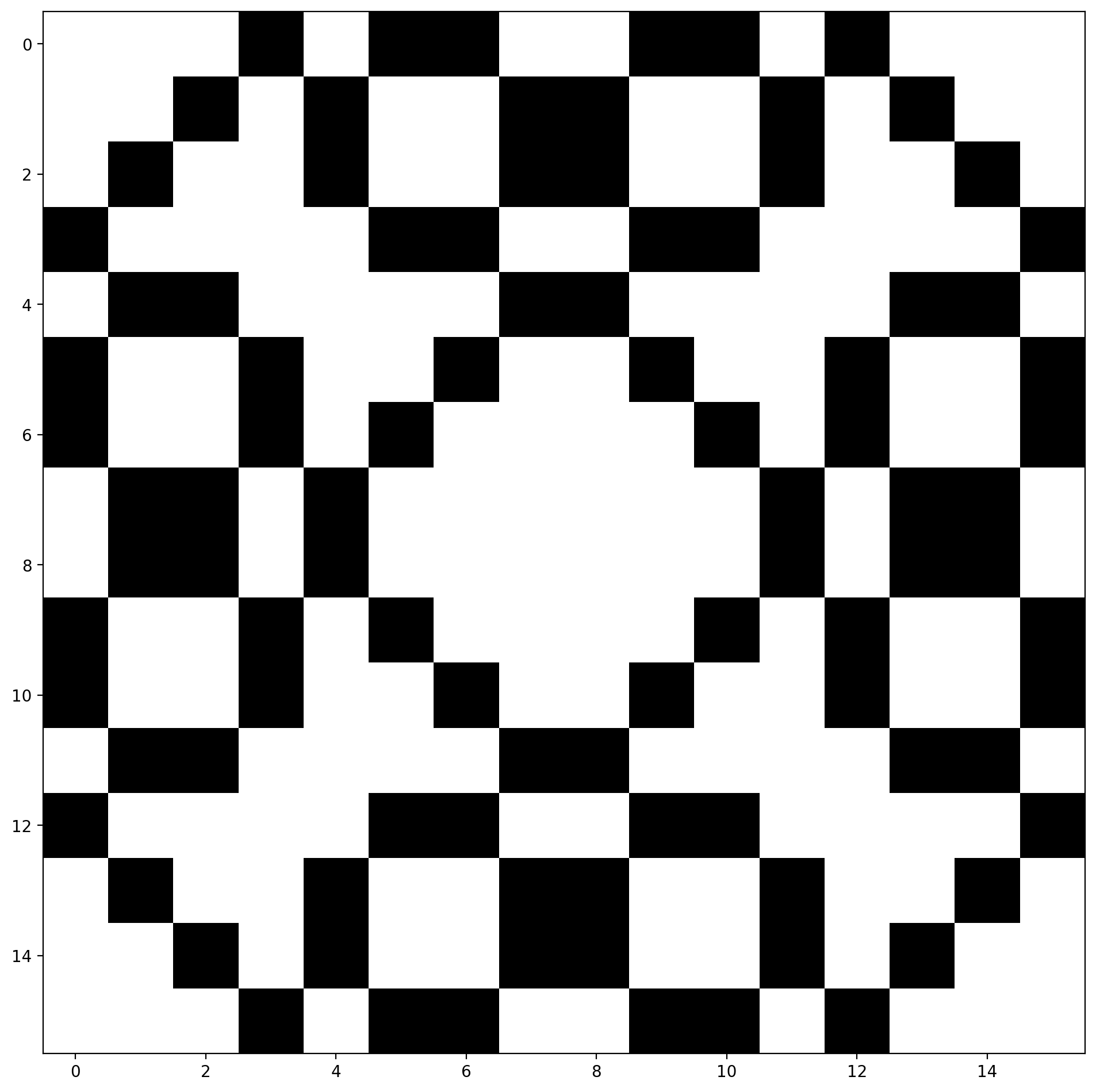}
        \caption{$T_{\Ham(2)}$, $|J|=2^4$.}
    \end{subfigure}
    \hfill
    \begin{subfigure}{0.24\textwidth}
        \includegraphics[width=\textwidth]{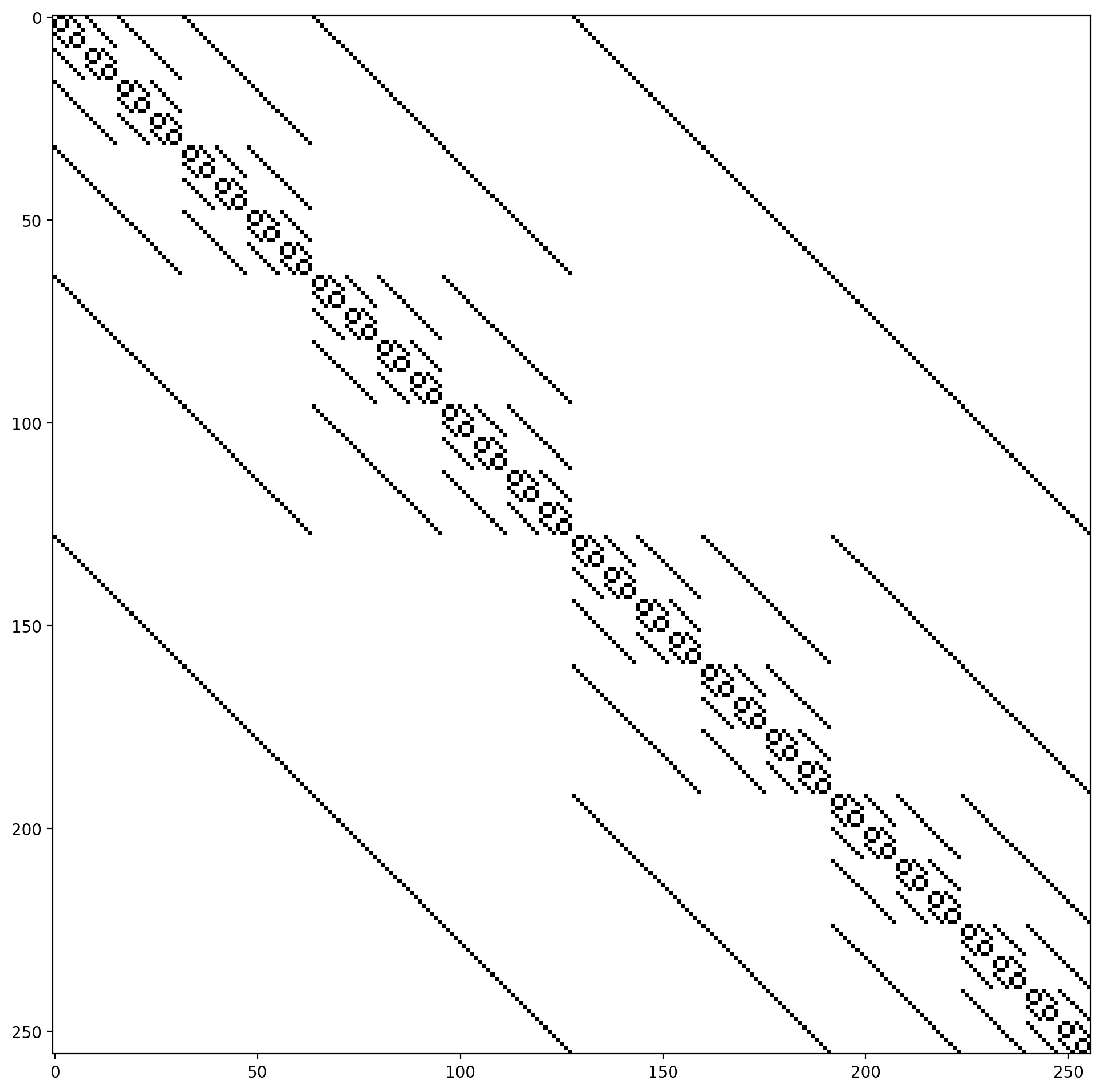}
        \caption{$T_{\Ham(1)}$, $|J|=2^8$.}
    \end{subfigure}
    \hfill
    \begin{subfigure}{0.24\textwidth}
        \includegraphics[width=\textwidth]{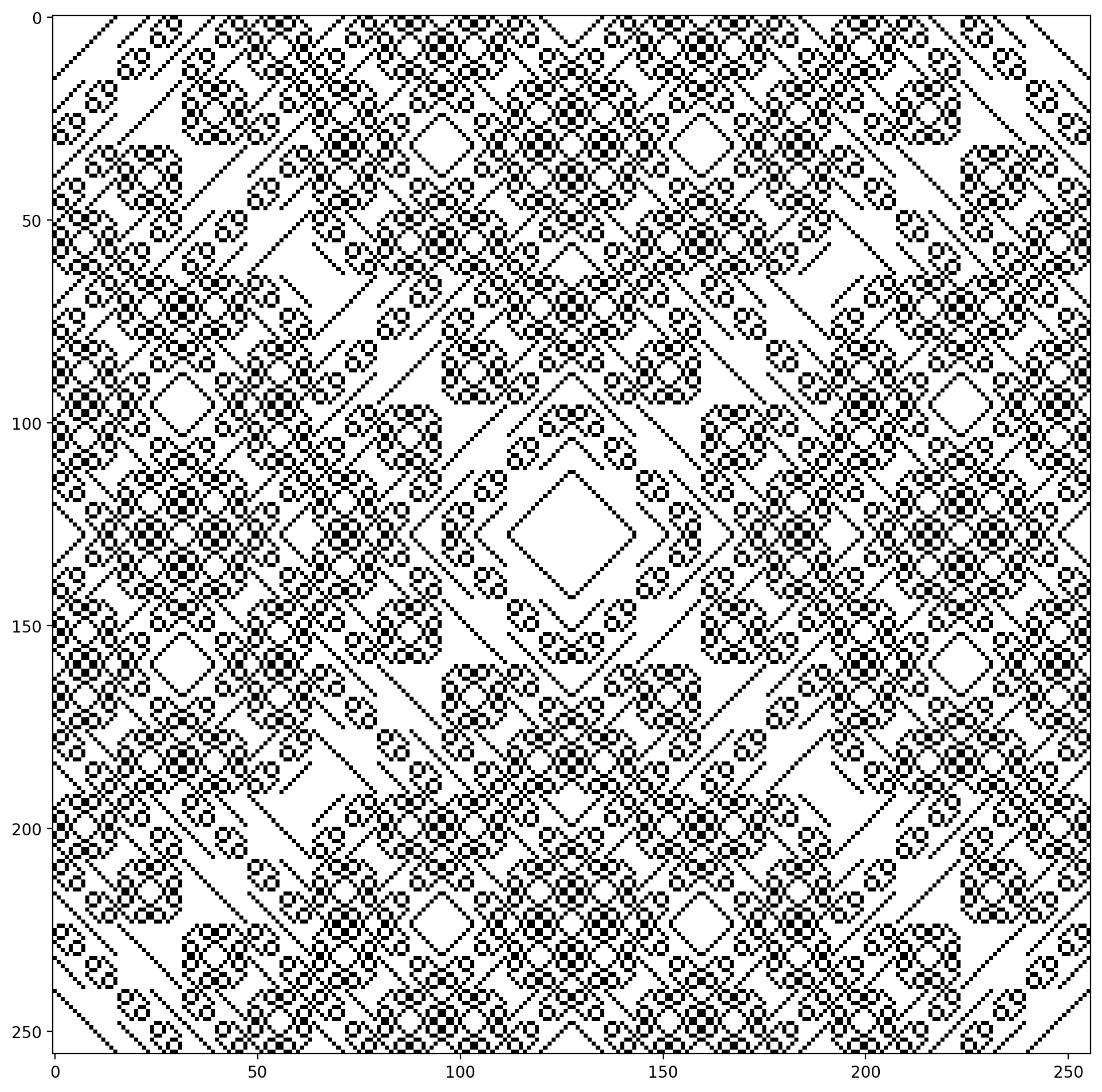}
        \caption{$T_{\Ham(4)}$, $|J|=2^8$.}
    \end{subfigure}
    \caption{Examples of the structure of $T_{\Ham(d)}$. The black color represents non-vanishing entries equal to one, representing pairs with the specified Hamming distance.}
    \label{fig:Tn48}
\end{figure}
Furthermore, it will be useful to denote the matrix which has two non\deleted{e}-zero entries at 
$(k,l)$ and $(l,k)$ as
\begin{equation}
    T_{k\leftrightarrow l}, \quad \text{ with } ( T_{k\leftrightarrow l})_{i,j}=\begin{cases}
    1, & \ \text{if } (i,j)=(k,l) \text{ or } (i,j)=(l,k),\\
    0, & \ \text{else.}
    \end{cases}
\end{equation}
Before we start, we point out that the diagonal entries of $T$ can be chosen to be zero, because $|(T^0)_{j,j}|=1\neq0$ for all $j\in J$.
Although trivial, we will repeatedly use that $v=\frac{1}{\sqrt{|J|}}( 1, 1, \cdots , 1)^T$ is an eigenvector of a matrix $F\in \C^{|J|\times |J|}$ if the sum of all rows are a multiple of $v$. 

\subsubsection{Hamming distance one mixer \texorpdfstring{$T_{\Ham(1)}$}{}}\label{sec:T_Ham1}
The matrix $T_{\Ham(1)}\in\R^{|J|\times |J|}$ fulfills Theorem~\ref{theorem:MixerHamiltonian} when $|J|=2^n, n \in \N$.
The symmetry of $T_{\Ham(1)})$ is due to the fact that the Hamming distance is a symmetric function.
Using the identity
\begin{equation}
        T_{\Ham(k)}T_{\Ham(1)}=T_{\Ham(1)}T_{\Ham(k)}=(n-(k-1))T_{\Ham(k-1)}+(k+1)T_{\Ham(k+1)}
\end{equation}
it can be shown that
\begin{equation}
T_{\Ham(1)}^k = \sum_{j=1}^{k-1} c_k T_{\Ham(j)} + k!T_{\Ham(k)},
\end{equation}
where $c_k$ are real coefficients.
Therefore, it is clear that $T_{\Ham(1)}^k$ reaches all states with Hamming distance $K$.
Furthermore, $v=\frac{1}{\sqrt{2^n}}( 1, 1, \cdots , 1)^T$ is a unit eigenvector of $T_{\Ham(1)}$ since the sum of each row is $n$. This is because there are exactly $n$ other states with Hamming distance one for each bitstring.

\subsubsection{All-to-all mixer \texorpdfstring{$T_A$}{}}\label{sec:T_A}
We denote the matrix with all but the diagonal entries equal to one as
\begin{equation}
    T_A, \quad \text{ with } ( T_A)_{i,j}=\begin{cases}
    1, & \ \text{if } i\neq j,\\
    0, & \ \text{else}.
    \end{cases}\\
\end{equation}
Trivially, $T_A\in\R^{|J|\times |J|}$ fulfilles Theorem~\ref{theorem:MixerHamiltonian} and $v=\frac{1}{\sqrt{|J|}}( 1, 1, \cdots , 1)^T$ is a unit eigenvector of $T_A $ since the sum of each row is $|J|-1$.

\subsubsection{(Cyclic) Nearest int\added{eger} mixer \texorpdfstring{$T_{\Delta}$/$T_{\Delta,c}$}{}}\label{sec:T_Delta}
Inspired by the stencil of finite-difference methods we introduce
$T_{\Delta}$, $T_{\Delta,c}\in\R^{|J|\times |J|}$ as \deleted{the} matrices with off-diagonal \added{entries} equal to one
\begin{equation}\label{eq:Tnearestint}
\begin{split}
    T_{\Delta},& \quad \text{ with } ( T_{\Delta})_{i,j}=\begin{cases}
    1, & \ \text{if } i=j+1 \vee i=j-1,\\
    0, & \ \text{else},
    \end{cases} \\
        T_{\Delta,c},& \quad \text{ with } ( T_{\Delta,c})_{i,j}=\begin{cases}
    1, & \ \text{if } i=(j+1)\operatorname{mod}n \vee i=(j-1)\operatorname{mod} n,\\
    0, & \ \text{else}.
    \end{cases}
    \end{split}
\end{equation}
Both matrices fulfill Theorem~\ref{theorem:MixerHamiltonian}. Symmetry holds by definition and it is easy to see that the k-th off-diagonal of $T_{\Delta}^k$ and $T_{\Delta,c}^k$ is nonzero for $1\leq k \leq |J|$.

For the nearest integer mixer $T_{\Delta}$ it is known that \begin{equation}
\label{eq:eigvecT_delta}
    v_k=(\sin(c),\allowbreak \sin(2c),\allowbreak \cdots,\allowbreak \sin(|J|c)), \ c=\frac{k \pi}{|J|+1}
\end{equation}
 are eigenvectors for $1\leq k \leq |J|$.
For the cyclic nearest integer mixer, we have that the sum of each row/column of $T_{\Delta,c}$ is equal to two (except for $n=1$ when it is one).
Therefore, $v=\frac{1}{\sqrt{|J|}}( 1, 1, \cdots , 1)^T$ is a unit eigenvector.

\begin{figure}
    \centering
    \bgroup
    \setlength\tabcolsep{-4pt}
    \begin{tabular}{>{\raggedright\arraybackslash}p{0.2\linewidth}
    >{\raggedright\arraybackslash}c@{\hspace{-32pt}}c@{\hspace{-32pt}}c@{\hspace{-32pt}}c@{\hspace{-32pt}}c@{\hspace{-32pt}}c@{\hspace{-28pt}}c@{\hspace{-28pt}}c@{\hspace{-28pt}}c@{\hspace{-28pt}}c@{\hspace{-28pt}}c}
        &5&6&7&8&9&10&11&12&13&14&15\\
         \scalebox{.75}{\shortstack[l]{$1\times$ $(T_1\!  +\!  T_2)$\\[1.2\baselineskip]}}
        &
         \raisebox{-.25\height}{\includegraphics[width=0.13\linewidth]{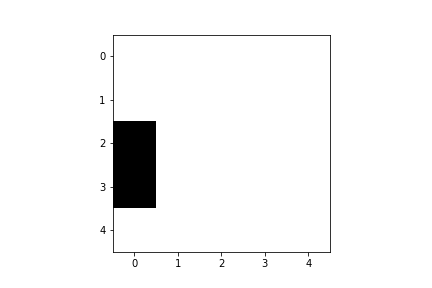} }&
         \raisebox{-.25\height}{\includegraphics[width=0.13\linewidth]{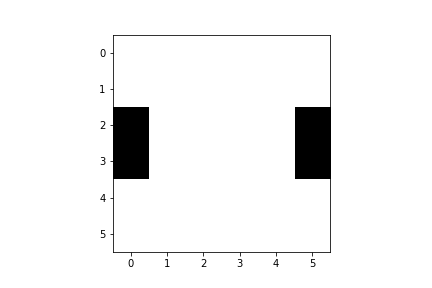} }&
         \raisebox{-.25\height}{\includegraphics[width=0.13\linewidth]{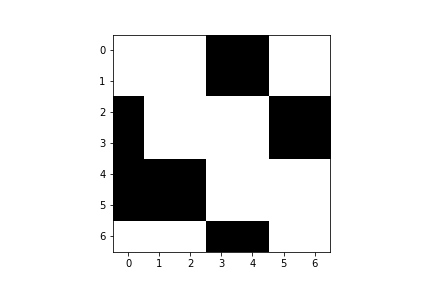} }&
         \raisebox{-.25\height}{\includegraphics[width=0.13\linewidth]{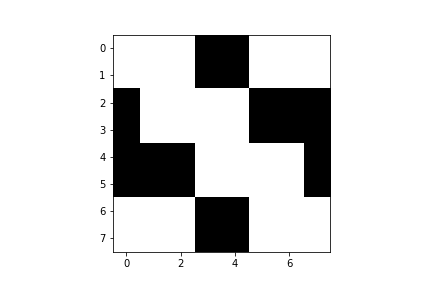} }&
         \raisebox{-.25\height}{\includegraphics[width=0.13\linewidth]{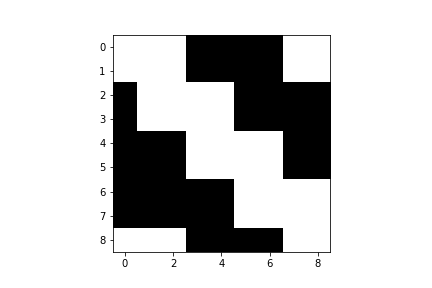} }&
         \raisebox{-.25\height}{\includegraphics[width=0.13\linewidth]{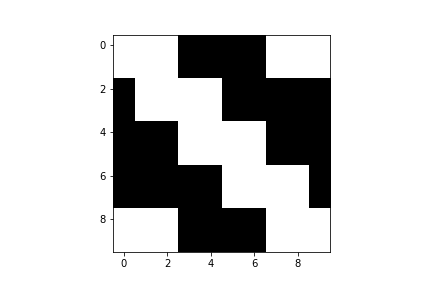}} &
         \raisebox{-.25\height}{\includegraphics[width=0.13\linewidth]{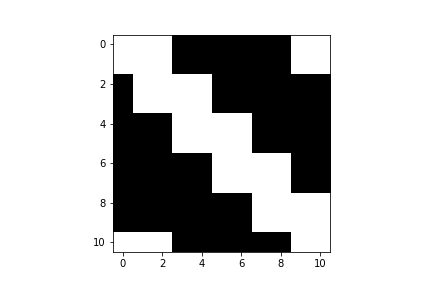}} &
         \raisebox{-.25\height}{\includegraphics[width=0.13\linewidth]{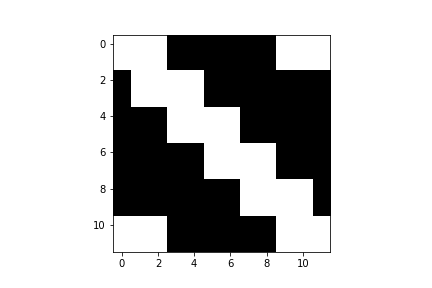}} &
         \raisebox{-.25\height}{\includegraphics[width=0.13\linewidth]{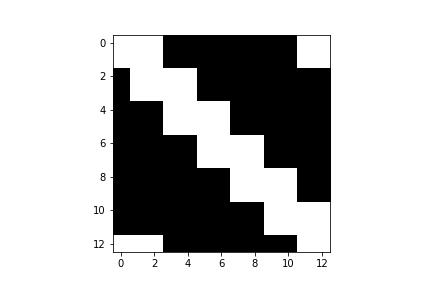}} &
         \raisebox{-.25\height}{\includegraphics[width=0.13\linewidth]{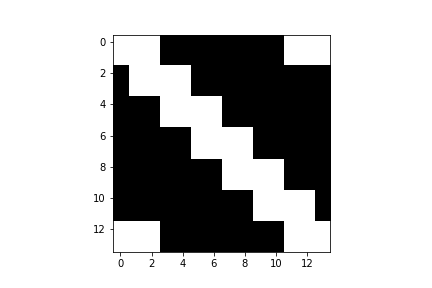}} &
         \raisebox{-.25\height}{\includegraphics[width=0.13\linewidth]{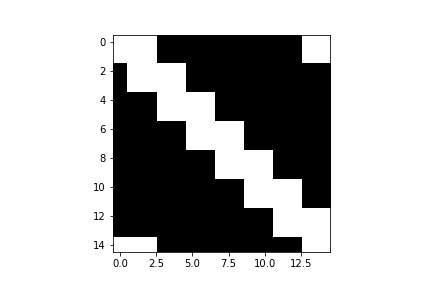}}\\[-4pt]
        \scalebox{.75}{\shortstack[l]{$2\times$ $(T_1\!  +\!  T_2)$\\[1.2\baselineskip]}}
         &
         \raisebox{-.25\height}{\includegraphics[width=0.13\linewidth]{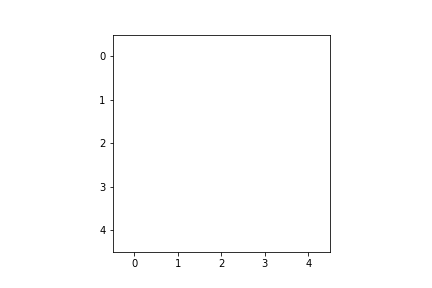} }&
         \raisebox{-.25\height}{\includegraphics[width=0.13\linewidth]{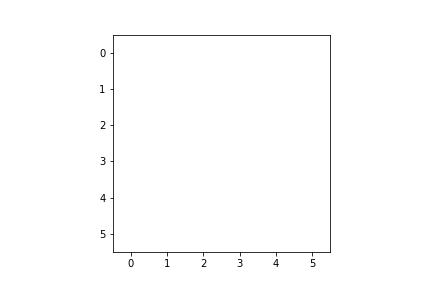} }&
         \raisebox{-.25\height}{\includegraphics[width=0.13\linewidth]{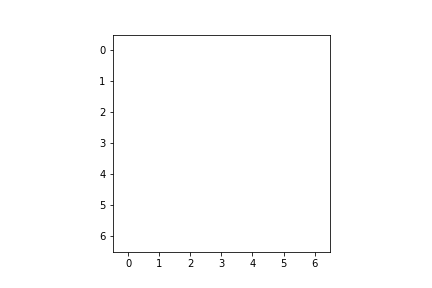} }&
         \raisebox{-.25\height}{\includegraphics[width=0.13\linewidth]{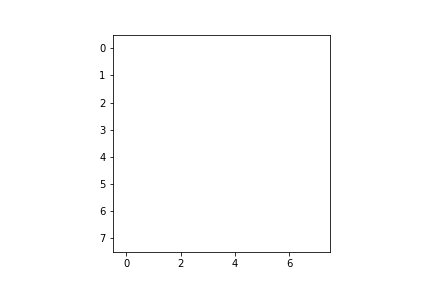} }&
         \raisebox{-.25\height}{\includegraphics[width=0.13\linewidth]{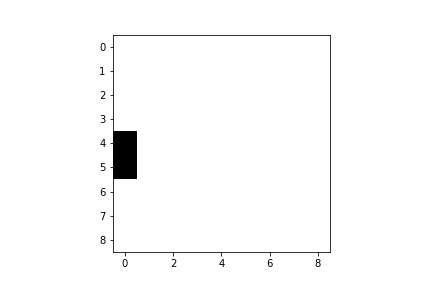} }&
         \raisebox{-.25\height}{\includegraphics[width=0.13\linewidth]{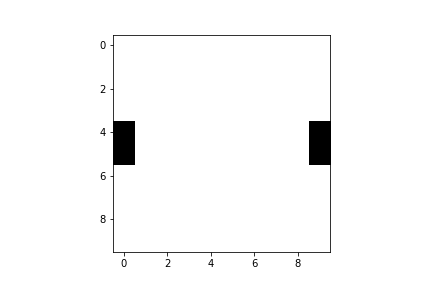}} &
         \raisebox{-.25\height}{\includegraphics[width=0.13\linewidth]{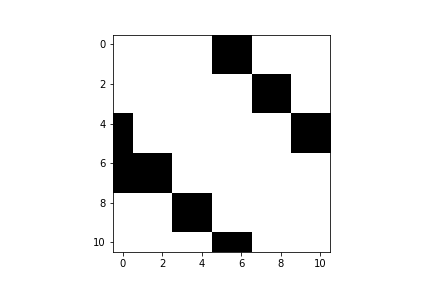}} &
         \raisebox{-.25\height}{\includegraphics[width=0.13\linewidth]{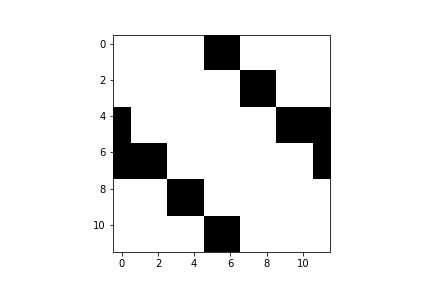}} &
         \raisebox{-.25\height}{\includegraphics[width=0.13\linewidth]{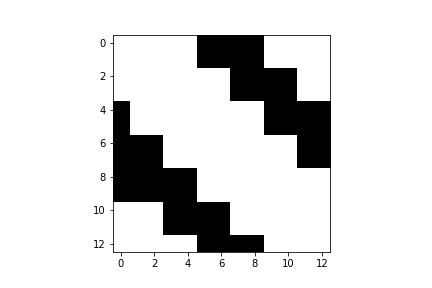}} &
         \raisebox{-.25\height}{\includegraphics[width=0.13\linewidth]{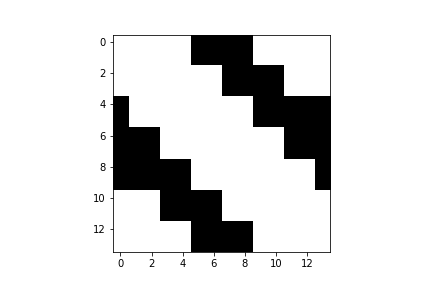}} &
         \raisebox{-.25\height}{\includegraphics[width=0.13\linewidth]{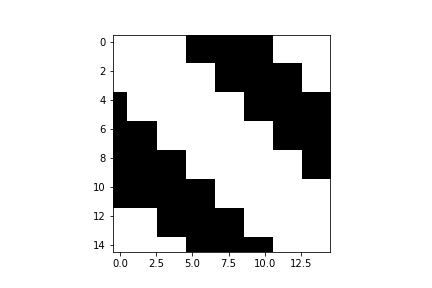}}\\[-4pt]
         \scalebox{.75}{\shortstack[l]{$3\times$ $(T_1\!  +\!  T_2)$\\[1.2\baselineskip]}}
         &
         \raisebox{-.25\height}{\includegraphics[width=0.13\linewidth]{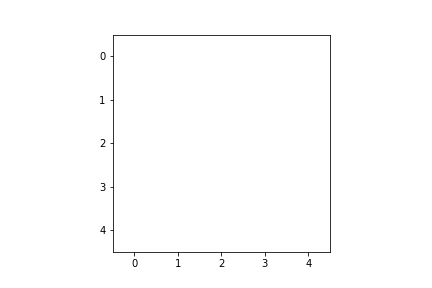} }&
         \raisebox{-.25\height}{\includegraphics[width=0.13\linewidth]{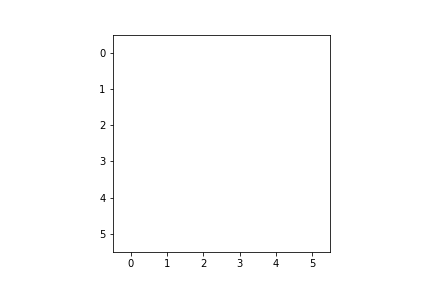} }&
         \raisebox{-.25\height}{\includegraphics[width=0.13\linewidth]{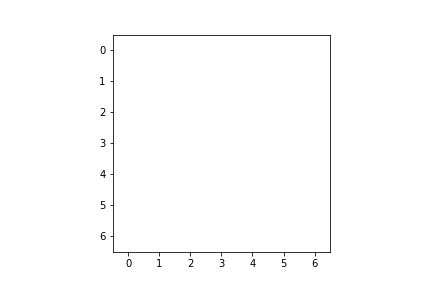} }&
         \raisebox{-.25\height}{\includegraphics[width=0.13\linewidth]{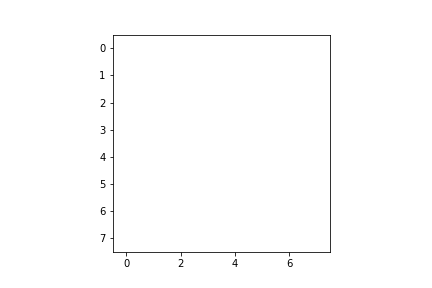} }&
         \raisebox{-.25\height}{\includegraphics[width=0.13\linewidth]{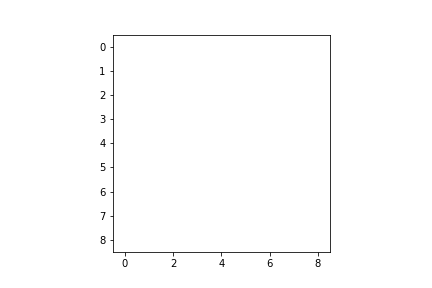} }&
         \raisebox{-.25\height}{\includegraphics[width=0.13\linewidth]{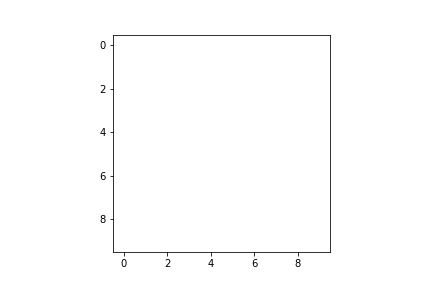}} &
         \raisebox{-.25\height}{\includegraphics[width=0.13\linewidth]{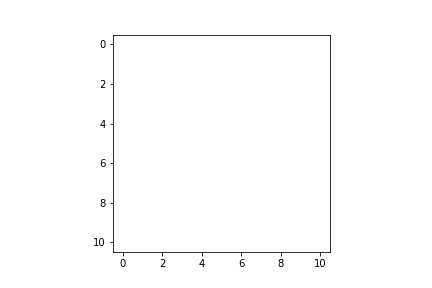}} &
         \raisebox{-.25\height}{\includegraphics[width=0.13\linewidth]{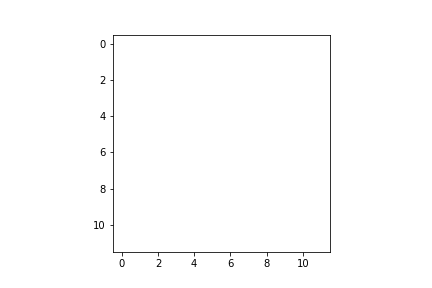}} &
         \raisebox{-.25\height}{\includegraphics[width=0.13\linewidth]{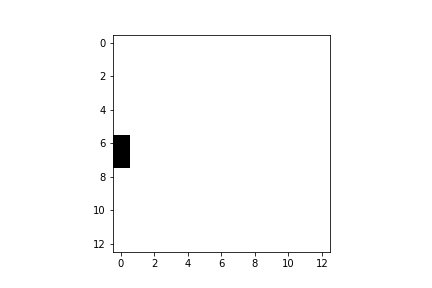}} &
         \raisebox{-.25\height}{\includegraphics[width=0.13\linewidth]{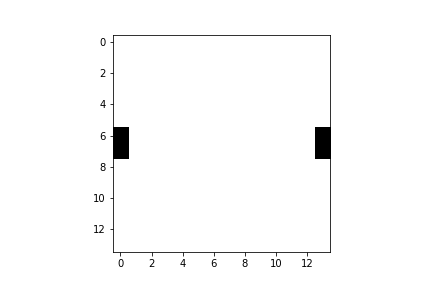}} &
         \raisebox{-.25\height}{\includegraphics[width=0.13\linewidth]{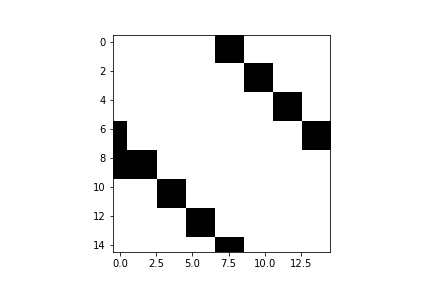}}\\[-4pt]
        \scalebox{.75}{\shortstack[l]{
        $(T_1\!  +\!  T_2)
        +T_{O(3)}\!  +\!  T_{E(3)}$\\[1.2\baselineskip]}}
         &
         \raisebox{-.25\height}{\includegraphics[width=0.13\linewidth]{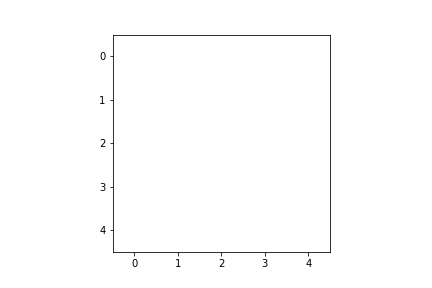} }&
         \raisebox{-.25\height}{\includegraphics[width=0.13\linewidth]{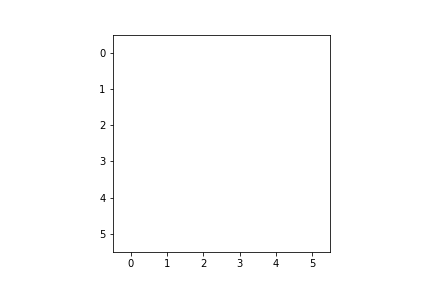} }&
         \raisebox{-.25\height}{\includegraphics[width=0.13\linewidth]{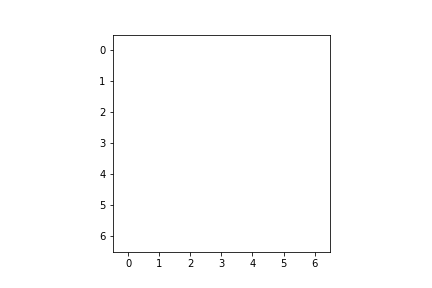} }&
         \raisebox{-.25\height}{\includegraphics[width=0.13\linewidth]{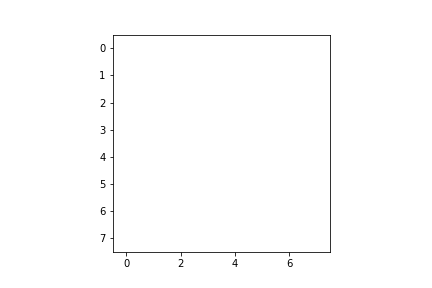} }&
         \raisebox{-.25\height}{\includegraphics[width=0.13\linewidth]{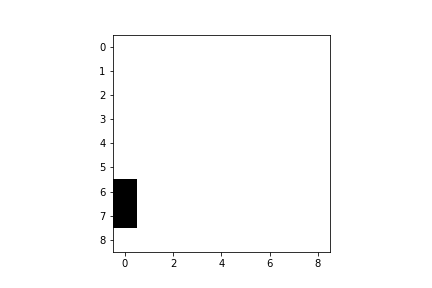} }&
         \raisebox{-.25\height}{\includegraphics[width=0.13\linewidth]{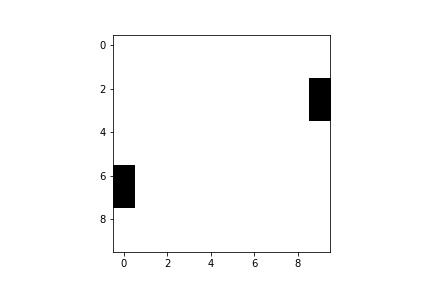}} &
         \raisebox{-.25\height}{\includegraphics[width=0.13\linewidth]{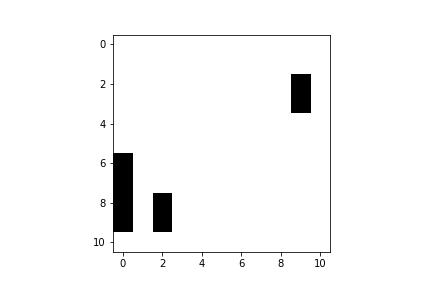}} &
         \raisebox{-.25\height}{\includegraphics[width=0.13\linewidth]{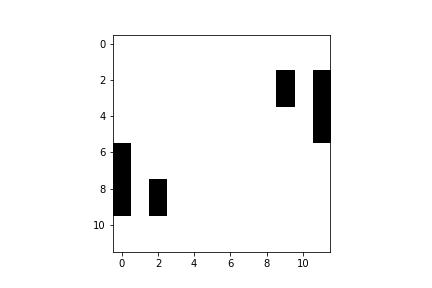}} &
         \raisebox{-.25\height}{\includegraphics[width=0.13\linewidth]{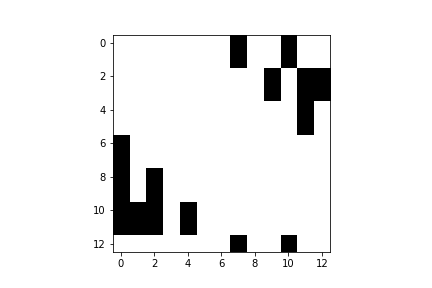}} &
         \raisebox{-.25\height}{\includegraphics[width=0.13\linewidth]{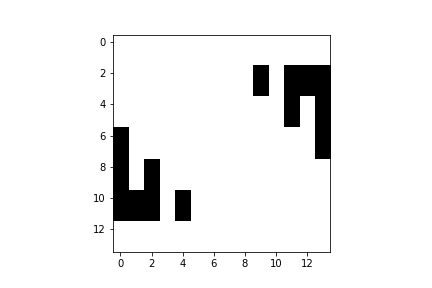}} &
         \raisebox{-.25\height}{\includegraphics[width=0.13\linewidth]{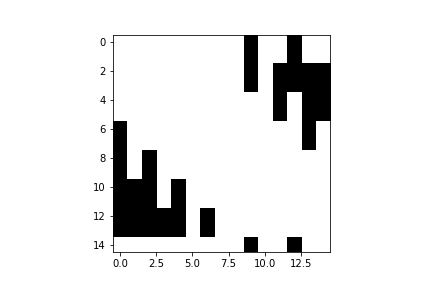}}\\[-4pt]
        \scalebox{.75}{\shortstack[l]{
        $(T_1\!  +\!  T_2)
        +T_{O(4)}\!  +\!  T_{E(4)}$\\[1.2\baselineskip]}}
         &
         &
         \raisebox{-.25\height}{\includegraphics[width=0.13\linewidth]{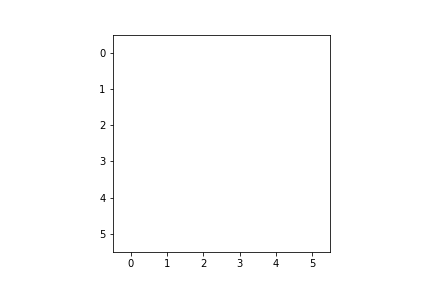} }&
         \raisebox{-.25\height}{\includegraphics[width=0.13\linewidth]{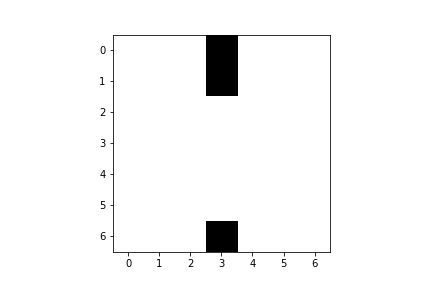} }&
         \raisebox{-.25\height}{\includegraphics[width=0.13\linewidth]{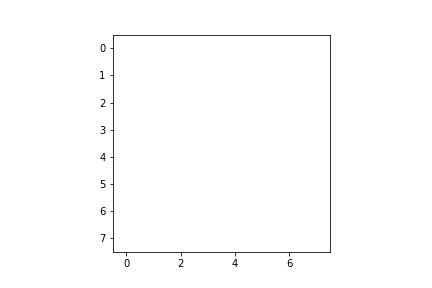} }&
         \raisebox{-.25\height}{\includegraphics[width=0.13\linewidth]{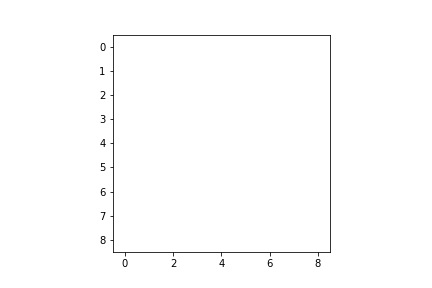} }&
         \raisebox{-.25\height}{\includegraphics[width=0.13\linewidth]{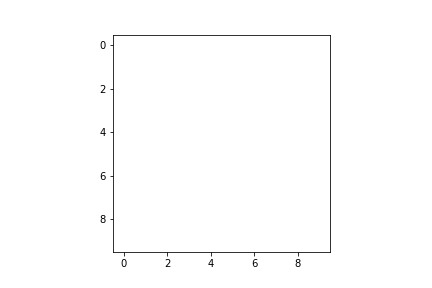}} &
         \raisebox{-.25\height}{\includegraphics[width=0.13\linewidth]{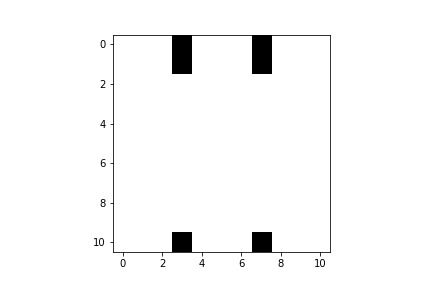}} &
         \raisebox{-.25\height}{\includegraphics[width=0.13\linewidth]{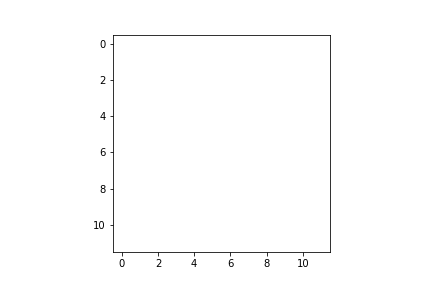}} &
         \raisebox{-.25\height}{\includegraphics[width=0.13\linewidth]{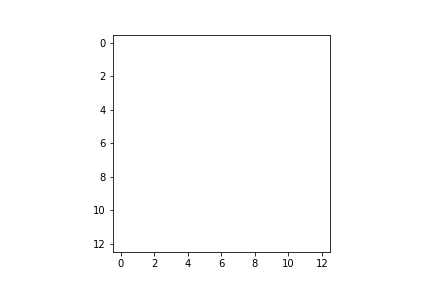}} &
         \raisebox{-.25\height}{\includegraphics[width=0.13\linewidth]{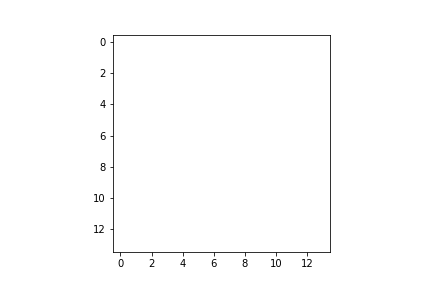}} &
         \raisebox{-.25\height}{\includegraphics[width=0.13\linewidth]{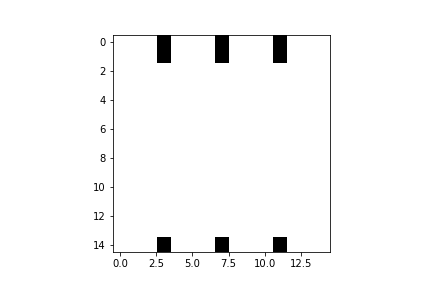}}\\[-4pt]
         \scalebox{.75}{
         \shortstack[l]{
        $(T_1\!  +\!  T_2)
        +\sum_{i\in I} (T_{O(i)}\!  +\!  T_{E(i)})$\\
         $I=\{3,5\},\{3,6\},\{5,6\},\{4,7\}$
         }}
         &
         &
         &
         &
         &
         \raisebox{-.25\height}{\includegraphics[width=0.13\linewidth]{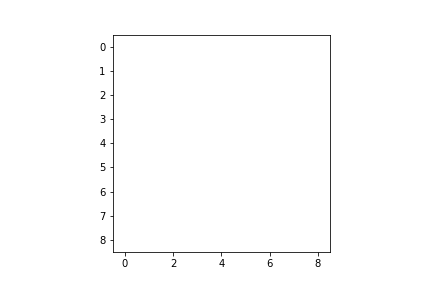} }&
         \raisebox{-.25\height}{\includegraphics[width=0.13\linewidth]{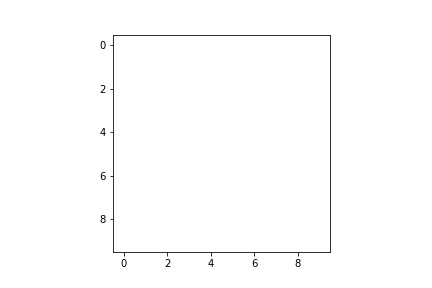}} &
         \raisebox{-.25\height}{\includegraphics[width=0.13\linewidth]{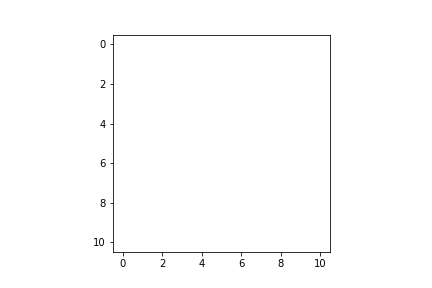}} &
         \raisebox{-.25\height}{\includegraphics[width=0.13\linewidth]{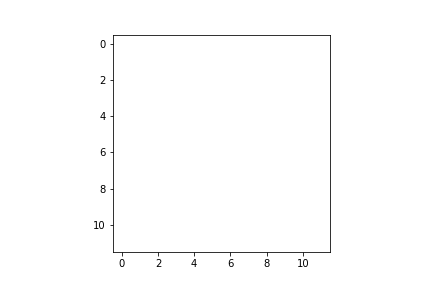}} &
         \raisebox{-.25\height}{\includegraphics[width=0.13\linewidth]{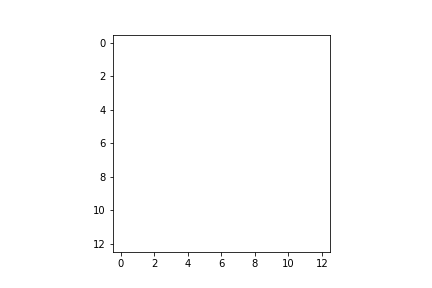}} &
         \raisebox{-.25\height}{\includegraphics[width=0.13\linewidth]{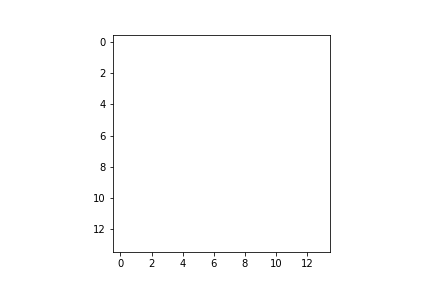}} &
         \raisebox{-.25\height}{\includegraphics[width=0.13\linewidth]{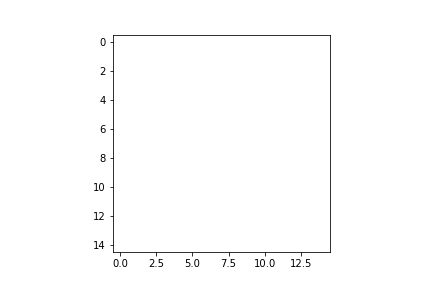}}\\[-4pt]
    \end{tabular}
    \egroup
    \caption{
    Valid (white) and invalid (black) transitions between pairs of states, 
    as defined in Theorem~\ref{theorem:TrotterizedMixerHamiltonian} for Trotterized mixer Hamiltonians.
    The first row shows that for 
    $T_1=T_{O(1),c}$ and $T_2=T_{E(1)}$, the mixer
    $U=e^{-i\beta T_1}e^{-i\beta T_2}$ does not provide transitions between all pairs of feasible states, although $U=e^{-i\beta (T_1+T_2)}$ does.
    }
    \label{fig:TOE_trotter}
\end{figure}

\subsubsection{Products of mixers and \texorpdfstring{$T_E,T_O$}{}}\label{sec:Toddeven}
In some cases, it will be necessary to use Theorem~\ref{theorem:TrotterizedMixerHamiltonian} to implement mixer unitaries. When splitting transitions matrices into odd and even entries the following definition is useful. \replaced{Denote}{Let} the matrix with entries in the d-th off-diagonal for even rows equal to one
\begin{equation}
    T_{E(d)}, \quad \text{ with } ( T_{E(d)})_{i,j}=\begin{cases}
    1, & \ \text{if } i=j+d \vee i=j-d, \text{ and } i \text{ even},\\
    0, & \ \text{else},
    \end{cases}
\end{equation}
and accordingly $T_{O(d)}$ for odd rows. In addition, we will \replaced{use}{need} $T_{O(1),c}$ to be the cyclic version in the same way as in Equation~\eqref{eq:Tnearestint}.
As an example, this allows \added{one} to decompose $T_{\Delta,c}=T_1+T_2\in\R^{n\times n}$ with $T_1=T_{O(1)}+T_{O(n-1)}=T_{O(1),c}$ and $T_2=T_{E(1)}$.

\subsubsection{Random mixer \texorpdfstring{$T_\text{rand}$}{}}
Finally, the upper triangular entries of the mixer $T_\text{rand}$ are drawn from a continuous uniform distribution on the interval $[0,1]$, and the lower triangular entries are chosen such that $T$ becomes symmetric.
Since the probability of getting a zero entry is zero, such a random mixer fulfilles Theorem~\ref{theorem:MixerHamiltonian} with probability 1.

\subsection{Decomposition of (constraint) mixers into basis gates}\label{sec:decomp}
Given a set of feasible (computational basis) states $B=\left\{\ket{x_j}, \ \ j\in J, \ \ x_j\in\{0,1\}^n \right\}$, we can use Theorem~\ref{theorem:MixerHamiltonian} to define a suitable mixer Hamiltonian.
The next question is how to (efficiently) decompose the resulting mixer into basis gates.
In order to do so we first decompose the Hamiltonian $H_M$ into a weighted sum of Pauli-strings.
A Pauli-string $P$ is a Hermitian operator of the form $P=P_1\otimes \cdots \otimes P_n$ where $P_i\in\{I,X,Y,Z\}$. Pauli-strings form a basis of the \textit{real} vector space of all $n$-qubit Hermitian operators.
%
Therefore, we can write
\begin{equation}\label{eq:Paulidecomp}
    H_M = \sum_{i_1,\cdots,i_n=1}^4 c_{i_1,\cdots,i_n} \ \ \sigma_{i_1}\otimes \cdots \otimes \sigma_{i_n}, \quad  c_{i_1,\cdots,i_n}\in\R,
\end{equation}
with real coefficients $c_{i_1,\cdots,i_n}$, 
where $\sigma_1=I, \sigma_2=X, \sigma_3=Y, \sigma_4=Z$.
After using a standard Trotterization scheme\added{\cite{hatano2005finding,trotter1959product}} (which is exact for commuting Paul-strings),
\begin{equation}
    U_M(t) = e^{-it H_M} \approx \prod_{\underset{|c_{i_1,\cdots,i_n|>0}}{i_1,\cdots,i_n=1}}^4 e^{-it c_{i_1,\cdots,i_n} \ \sigma_{i_1}\otimes \cdots \otimes \sigma_{i_n}},
    \label{eq:Trotterization}
\end{equation}
it is well-established how to implement each of the terms of the product using basis gates, see Equation~\eqref{eq:eP}.
We will discuss the effects of Trotterization in more detail in Section~\ref{sec:Trotterizations}, as there are several important aspects to consider for a valid mixer.
\begin{equation}
    e^{-i t P} = 
    \begin{quantikz}[row sep={0.5cm,between origins},column sep=1ex]
    \qw&\gate[wires=7]{U}
    & \ctrl{1} & \qw & \qw &\qw& \qw& \qw& \qw & \qw & \qw &\qw& \qw& \qw& \qw& \ctrl{1} & \gate[wires=7]{U^\dagger}& \qw\\
    \qw&\qw& \targ{} &\ctrl{1} & \qw & \qw& \qw& \qw& \qw& \qw & \qw & \qw& \qw& \qw& \ctrl{1} &\targ{}&\qw& \qw\\
    \qw& \qw& \qw \qw & \targ{} & \qw & \qw & \qw& \qw& \qw& \qw & \qw & \qw& \qw& \qw&\targ{} & \qw & \qw & \qw\\
    &\qwbundle[alternate]{}&\qwbundle[alternate]{} &&\ddots&&&&&&&\adots&&&&&\qwbundle[alternate]{}&\qwbundle[alternate]{}\\
    \qw&\qw& \qw & \qw & \qw & \qw & \ctrl{1} & \qw & \qw & \qw & \ctrl{1}& \qw & \qw & \qw & \qw & \qw & \qw& \qw \\
    \qw&\qw& \qw & \qw & \qw & \qw & \targ{} & \ctrl{1} & \qw& \ctrl{1} &\targ{} & \qw & \qw & \qw & \qw & \qw & \qw& \qw \\
    \qw&\qw& \qw & \qw & \qw & \qw & \qw & \targ{} & \gate{R_z(2t)} & \targ{}& \qw & \qw & \qw & \qw & \qw & \qw & \qw & \qw \\
    \end{quantikz},
    \label{eq:eP}
    \end{equation}
    \begin{equation*}
    U_i=\begin{cases}
    H,& \text{ if } P_i=X,\\
    SH,& \text{ if } P_i=Y,\\
    I,& \text{ if } P_i=Z,
\end{cases}
\qquad 
    (U^\dagger)_i=\begin{cases}
    H,& \text{ if } P_i=X,\\
    HS^\dagger,& \text{ if } P_i=Y,\\
    I,& \text{ if } P_i=Z.
\end{cases}
\end{equation*}
\added{Here, $S$ is the S or Phase gate and $H$ is the Hadamard gate.}
The standard way to compute the coefficients $c_{i_1,\cdots,i_n}$ is given in Algorithm~\ref{alg:PauliDecompStandard}.
\begin{algorithm}[ht]
\KwData{Feasible states $B$ and transition matrix T fulfilling Theorem~\ref{theorem:MixerHamiltonian}}
\KwResult{Coefficients $c_{i_1,\cdots,i_n}$ of Pauli-strings, Equation~\eqref{eq:Paulidecomp}}
    \For{$i_1=1,\ldots,4$}{
    
        $\ddots$
        
\bgroup
\parindent 5pt
\leftskip 5pt 
\SetInd{17pt}{17pt}
            \For{$i_n=1,\ldots,4$}{
                $
                    c_{i_1,\cdots,i_n} = \frac{1}{2^n} \Trace{ \sigma_{i_1}\otimes \cdots \otimes \sigma_{i_n} H_M }
                $
            }
\egroup
        
        \reflectbox{$\ddots$}
        
    }
\caption{Decompose $H_M$ given by Equation~\eqref{eq:Hmdefinition} into Pauli-strings via trace}\label{alg:PauliDecompStandard}
\end{algorithm}
%
%
For $n$ qubits this requires to compute $4^n$ coefficients, as well as multiplication of $2^n\times 2^n$ matrices.
However, most of these terms are expected to vanish.
We therefore describe an alternative way to produce this decomposition, using the language of quantum \replaced{mechanics~\cite{sakurai2006advanced}}{chemistry}.
In the following we use the \textit{ladder operators}
used in the 
creation and annihilation operators 
from the second quantization formulation in quantum chemistry defined by
\begin{equation}
    a^\dagger = \frac{1}{2} (X-iY), 
    a = \frac{1}{2} (X+iY).
\end{equation}
Since 
$a\ket{0} = 0, a\ket{1} = \ket{0}$,
where $0$ is the zero vector,
we have that $\ket{0}\bra{1}=a$.
Since
$a^\dagger\ket{0} = \ket{1}, a^\dagger\ket{1} = 0$,
we have that $\ket{1}\bra{0}=a^\dagger$,
and finally 
$    a^\dagger a \ket{0} = 0, a^\dagger a \ket{1} = \ket{1}, 
    a a^\dagger \ket{0} = \ket{0}, a a^\dagger \ket{1} = 0,
$
means that  $\ket{0}\bra{0}=aa^\dagger$ and  $\ket{1}\bra{1}=a^\dagger a$.
Note that
\begin{equation}
    a^\dagger a = \frac{1}{2}\left(I-Z\right), \quad 
    a a^\dagger = \frac{1}{2}\left(I+Z\right).
\end{equation}


    
                

As an example, consider the matrix $M = \ket{01}\bra{10}=\ket{0}\bra{1}\otimes\ket{1}\bra{0}$ which can be expressed with ladder
operators as $M= a_1 a_2^\dagger$.
Another example is given by $M = \ket{01}\bra{11} = a_1 a_2^\dagger a_2$.
This approach clearly extends to the general case and leads to \deleted{the} Algorithm~\ref{alg:PauliDecompSecondQ}.
\begin{algorithm}[ht]
\KwData{Feasible states $B$ and transition matrix T fulfilling Theorem~\ref{theorem:MixerHamiltonian}}
\KwResult{Non-vanishing coefficients $c_{i_1,\cdots,i_n}$ of Pauli-strings, Equation~\eqref{eq:Paulidecomp}}
    S=0
    
    \For{{\normalfont \textbf{all}} $1\leq j,k\leq |J|$, s.t.  $(T)_{j,k}\neq 0$}{
                Define $x_L=x_{J_j}, x_R=x_{J_k}$
                
                \For{$1\leq i \leq n$}{
                    $(P_{j,k})_i = \begin{cases}
                        \frac{1}{2} (X+iY),& \text{if } (x_L)_i=0 \wedge (x_R)_i=1,\\
                        \frac{1}{2} (X-iY),& \text{if } (x_L)_i=1 \wedge (x_R)_i=0,\\
                        \frac{1}{2}\left(I+Z\right),& \text{if } (x_L)_i=0 \wedge (x_R)_i=0,\\
                        \frac{1}{2}\left(I-Z\right),& \text{if } (x_L)_i=1 \wedge (x_R)_i=1.\\
                    \end{cases}$
                }
                $ S = S + (T)_{j,k} P_{j,k} $
            }
    simplify S (e.g., using a library for symbolic mathematics)
    
    this defines the non-vanishing coefficients $c_{i_1,\cdots,i_n}$
\caption{Decompose $H_M$ given by Equation~\eqref{eq:Hmdefinition} into Pauli-strings directly}\label{alg:PauliDecompSecondQ}
\end{algorithm}

A comparison of the complexity of the two algorithms is given in Table~\ref{tab:complexity}.
The naive algorithm needs to perform a matrix-matrix\replaced{ }{-}multiplication with matrices of \deleted{the} size $2^n\times2^n$ for each of the $4^n$ coefficients.
This quickly becomes prohibitive for larger $n$.
The algorithm based on ladder
operators requires resources that scale with the number of non-zero entries of the transition matrix $T$, which is much more favourable.
In the end a symbolic mathematics library is used to simplify the expressions in order to create the list of non-zero Pauli-strings.

\subsection{Optimality of mixers}\label{sec:optimality}
On current NISQ devices, the noise level of two-qubit gate (CX) times and error rates are one order of magnitude higher than for single qubit gates ($U_3$). In addition, most devices lack all-to-all connectivity.
CX gates between these require SWAP operations, which consist of  additional CX gates.
\textit{An optimal mixer will therefore contain as few CX gates as possible.}
Since Pauli-strings are implemented according to Equation~\eqref{eq:eP} we define the cost to implement $e^{-i t H_M}$ as
\begin{equation}
    \operatorname{Cost}(H_M) = \sum_{\substack{i_1,\cdots,i_n=1\\
    |c_{i_1,\cdots,i_n}|> 0\\
    \operatorname{len}(\sigma_{i_1}\otimes \cdots \otimes \sigma_{i_n})>1
    }}^4
    2 (\operatorname{len}(\sigma_{i_1}\otimes \cdots \otimes \sigma_{i_n})-1),
    \label{eq:optimality}
\end{equation}
where \replaced{$\operatorname{len}(P)$ is }{ we defined } the \textit{length of a Pauli-string P} 
\replaced{defined as }{to be} the number of literals that are not the identity.
For instance $P=IXIIY=I_1 X_2 I_3 I_4 Y_5=X_2 Y_5$ has $\operatorname{len}(P)=2$.
The $\operatorname{Cost}(H_M)$ specifies the number of CX gates that are required to implement the mixer. A lower cost means fewer and/or shorter Pauli-strings.
There are four interconnected factors that influence the cost to implement the mixer for a given $B$.

\subsubsection{Transition matrix \texorpdfstring{$T$}{T}}
The larger $|B|$ the more freedom we have in choosing the transition matrix $T$ that fulfills Theorem~\ref{theorem:MixerHamiltonian}. The combination of $T$ and the specific states of $B$ define the cost of the Hamiltonian. Unless one can find a way to utilize\deleted{s} the structure of the states of $B$ to efficiently compute an optimal $T$, we expect this problem to be NP-hard. In practice, a careful analysis of the specific states of $B$ is required to determine $T$ such that the cost becomes low.
We will revisit optimality for both unrestricted and restricted mixers in \replaced{ Sections~\ref{sec:fullmixer} and \ref{sec:constrainedmixers}.}{with}

\subsubsection{Adding mixers}
Corollary~\ref{corollary:properties} allows \added{one} to add mixers with a kernel that contains $\Span{B}$. In general, also this is a combinatorial optimization problem which we do not expect to solve exactly with an efficient algorithm. However, we will provide a heuristic that can be used to reduce the cost of mixers in certain cases. We will provide more details in Section~\ref{sec:constrainedmixers} where we discuss constrained mixers on some examples in detail.

\subsubsection{Non-commuting Pauli-strings}
Depending on the mixer -- which depends on the transition matrix and addition of mixers outside the feasible subspace -- one can influence the commutativity pattern of the resulting Pauli-strings. This is an intricate topic, which we discuss next.

\subsection{Trotterizations}\label{sec:Trotterizations}
Algorithms~\ref{alg:PauliDecompStandard} and \ref{alg:PauliDecompSecondQ} produce a weighted sum of Pauli-strings equal to the mixer Hamiltonian $H_M$ defined in Theorem~\ref{theorem:MixerHamiltonian}.
A further complication arises when the non-vanishing Pauli-strings of the mixer Hamiltonian $H_M$ do not all commute.
In that case one can not realize $U_M$ exactly but has to find a suitable approximation/Trotterization, see Equation~\eqref{eq:Trotterization}.
Two Pauli-strings \textit{commute}, i.e., $[P_A,P_B]=P_AP_B-P_BP_A=0$ if\added{,} and only if\added{,} they \textit{fail} to commute on an \textit{even} number of indices~\cite{gokhale2019minimizing}.
An example is given in Figure~\ref{fig:XYn3}.

This problem is similar to a problem for observables; how does one divide the Pauli-strings into groups of commuting families~\cite{gokhale2019minimizing,gui2020term} \deleted{such as} to maximize efficiency and increase accuracy?
In order to minimize the number of measurements required to estimate a given observable one wants to find a ``min-commuting-partition"; given a set of Pauli-strings from
a Hamiltonian, one seeks to partition the strings into commuting families such that the total number of partitions is minimized. This problem is NP-hard in general~\cite{gokhale2019minimizing}.
However, based on Theorem~\ref{theorem:P_Tij_commutes} we expect our problem to be much more tractable.

For our case, it turns out that \textbf{not all Trotterizations are suitable as mixing operators}; they can either fail to preserve the feasible subspace, i.e., Equation~\eqref{eq:mixer_preserve}, or fail to provide transitions between all pairs of feasible states, i.e., Equation~\eqref{eq:mixer_transition}.
An example is given by $B=\{\ket{001}, \ket{010}, \ket{100}\}$ with the mixer $H_M=\frac{1}{2} (XIX+YIY)+ \frac{1}{2} (XXI+YYI)$ associated with $T_{\Delta} = T_{1\leftrightarrow 2}+T_{2\leftrightarrow 3}$, see Section~\ref{sec:XYmixer}. Looking at Figure~\ref{fig:XYn3}, these terms can be grouped into commuting families in two ways, which represent two (of many) different ways to realize the mixer unitary with basis gates.
\begin{enumerate}
    \item The first possible Trotterization is given by $U_1(\beta)=e^{-i\beta(XXI+IXX)}$ and $U_2(\beta)=\allowbreak e^{-i\beta(YYI+IYY)}$.
    However, it turns out that $\exists \beta\in\R$ such that $|\bra{111}\allowbreak U_1(\beta)\allowbreak U_2(\beta)\allowbreak \ket{z}|>0$ for all $\ket{z}\in B$.
    This means that this Trotterization does \emph{not preserve the feasible subspace} and does not represent a valid mixer Hamiltonian.
    The underlying reason for this is that the terms $XXI$ and $YYI$ are generated from the entry $T_{1\leftrightarrow 2}$, but are split in this Trotterization. The same holds true for $IXX$ and $IYY$ which are generated via $T_{2\leftrightarrow 3}$.
    
    \item The second possible Trotterization is given by $U_1(\beta)=e^{-i\beta(XIX+YIY)}$ and $U_2(\beta)=e^{-i\beta(XXI+YYI)}$,
    which splits terms with respect to $T_{1\leftrightarrow 2}$ and $T_{2\leftrightarrow 3}$
    In this case, we have that $|\bra{100}\allowbreak U_1(\beta)\allowbreak U_2(\beta)\allowbreak \ket{001}|=0$, so it does \emph{not provide an overlap} between all feasible computational basis states.
    This can be understood via Theorem~\ref{theorem:TrotterizedMixerHamiltonian}.
    We have that $(T_{1\leftrightarrow 2}^{n_1}T_{2\leftrightarrow 3}^{n_2})_{3,1}=0$ for all $n_1,n_2\in\N$, so one can not ``reach" $\ket{100}$ from $\ket{001}$.
    The opposite is \emph{not} true; we have that $(T_{1\leftrightarrow 2}T_{2\leftrightarrow 3})_{1,3}=1$, so $\exists \beta$ such that 
    $|\bra{001}\allowbreak U_1(\beta)\allowbreak U_2(\beta)\allowbreak \ket{100}|>0$.
\end{enumerate}

We have just learned that it is a bad idea to Trotterize terms that belong to a non-zero entry of T, i.e. to $T_{j\leftrightarrow i}$.
Therefore, we need to \replaced{show that}{exclude that there are cases where not} all non-vanishing Pauli-strings of $\ket{x_j}\bra{x_i} + \ket{x_i}\bra{x_j} $ commute\replaced{; otherwise there might exist subspaces for which we can not realize the mixer constructed in Theorem~\ref{theorem:MixerHamiltonian}}{. If that were the case, it might not be possible to realize a mixer for a certain feasible subspace}.
Luckily, the following Theorem shows that it is always possible to realize a mixer by Trotterizing according to non-zero entries of $T=\sum_{i,j\in J, i<j} T_{j\leftrightarrow i}$.
\begin{theorem}[Pauli strings for $T_{j\leftrightarrow i}$ commute]\label{theorem:P_Tij_commutes}
Let $\ket{z}, \ket{w}$ be two computational basis states in $\C^{2^n}$.
Then all non-vanishing Pauli-strings $\sigma_{i_1}\otimes \cdots \otimes \sigma_{i_n}$ of the decomposition
\begin{equation}
    \ket{z}\bra{w} + \ket{w}\bra{z} = \sum_{i_1,\cdots,i_n=1}^4 c_{i_1,\cdots,i_n} \ \ \sigma_{i_1}\otimes \cdots \otimes \sigma_{i_n}, \quad  c_{i_1,\cdots,i_n}\in\R,
\end{equation}
commute.
\end{theorem}

\begin{proof}
We will \replaced{prove}{proof} the following more general assertion by induction.
Let $P_{1+,1},P_{1+,2}$ \added{be} two non-vanishing Pauli-strings of the decomposition of $\ket{z}\bra{w} + \ket{w}\bra{z}$, and 
$P_{i-,1},P_{i-,2}$ \added{be} two non-vanishing Pauli-strings of the decomposition of $i(\ket{z}\bra{w} - \ket{w}\bra{z})$. Then $[P_{1+,1},\allowbreak P_{1+,2}]=0$, $[P_{i-,1},P_{i-,2}]=0$ and $[P_{1+,\cdot},P_{i-,\cdot}]\neq0$.
We will use that two Pauli-strings commute if\added{,} and only if\added{,} they \textit{fail} to commute on an \textit{even} number of indices~\cite{gokhale2019minimizing}.

\noindent
\ul{For $n=1$} we have the following cases.
\begin{equation}
        A_1={\scriptstyle \ket{z}\bra{w} + \ket{w}\bra{z}}=\begin{cases}
        I+Z,& \! \! \! \! \text{if } \scriptstyle(z,w)=(0,0),\\
        X,& \! \! \! \!  \text{if } \scriptstyle(z,w)=(0,1),\\
        X,& \! \! \! \!  \text{if } \scriptstyle(z,w)=(1,0),\\
        I-Z,& \! \! \! \!  \text{if } \scriptstyle(z,w)=(1,1),\\
    \end{cases}
    \ \
        B_1={\scriptstyle i(\ket{z}\bra{w} - \ket{w}\bra{z})}=
        \begin{cases}
        0,& \! \! \! \!  \text{if} \scriptstyle(z,w)=(0,0),\\
        -Y,& \! \! \! \!  \text{if} \scriptstyle(z,w)=(0,1),\\
        +Y,& \! \! \! \!  \text{if} \scriptstyle(z,w)=(1,0),\\
        0,& \! \! \! \!  \text{if} \scriptstyle(z,w)=(1,1).\\
    \end{cases}
\end{equation}
It is trivially true that $[P_{1+,1},P_{1+,2}]=0$\replaced{ and}{,} $[P_{i-,1},P_{i-,2}]=0$ 
since the maximum number of Pauli-strings is two, and in that case one of the Pauli-strings is the identity.
Moreover, $P_{i-,\cdot}$ is nonzero only when $z\neq w$. In that case $[P_{1+,\cdot},P_{i-,\cdot}]=[X,\pm Y]\neq0$.

\noindent
\ul{$n\rightarrow n+1$.} The assumptions hold for $A_n= (\ket{z}\bra{w} + \ket{w}\bra{z}), \allowbreak B_n=i(\ket{z}\bra{w} - \ket{w}\bra{z})$. 
We have the following four cases
\begin{equation}
    \begin{split}
        A_{n+1}=\ket{zx}\bra{wy} + \ket{wy}\bra{zx}=&\frac{1}{2}\begin{cases}
        A_n\otimes(I+Z),& \text{if } (x,y)=(0,0),\\
        A_n\otimes X +B_n\otimes Y,& \text{if } (x,y)=(0,1),\\
        A_n\otimes X -B_n\otimes Y,& \text{if } (x,y)=(1,0),\\
        A_n\otimes(I-Z),& \text{if } (x,y)=(1,1),\\
    \end{cases}
    \\
        B_{n+1}=i(\ket{zx}\bra{wy} - \ket{wy}\bra{zx})=&
        \frac{1}{2}\begin{cases}
        B_n\otimes(I+Z),& \text{if } (x,y)=(0,0),\\
        B_n\otimes X-A_n\otimes Y,& \text{if } (x,y)=(0,1),\\
        B_n\otimes X+A_n\otimes Y,& \text{if } (x,y)=(1,0),\\
        B_n\otimes(I-Z),& \text{if } (x,y)=(1,1).\\
    \end{cases}
    \end{split}
    \label{eq:recursionPauli}
\end{equation}

\noindent
\textit{Case $x=y$.}
The number of indices where two Pauli-strings commute does not change when going from $A_n$ to $A_n\otimes(I\pm Z)$. The same holds for $B_n$ and $B_n\otimes(I\pm Z)$.
This means that the assertion is true for $x=y$.

\noindent
\textit{Case $x\neq y$.}
First, we prove that all non-vanishing Pauli-strings of $A_{n+1}$ commute, and the same for $B_{n+1}$.
This is easy to see, since non-vanishing Pauli strings of  $A_n\otimes X \pm B_n \otimes Y$ must have an even number of indices where they fail to commute. The same is true for $B_n\otimes X \pm A_n \otimes Y$.
Finally, we prove that non-vanishing Pauli-strings of $A_{n+1}$ do not commute with non-vanishing Pauli-strings of $B_{n+1}$.
Using our assumptions of non-vanishing Pauli strings of $A_n$ and $B_n$ it is easy to show that non-vanishing Pauli-strings of
the following pairs fail to commute on an \textit{odd} number of indices
$(A_n\otimes X,B_n\otimes X)$, 
$(A_n\otimes X,\pm A_n\otimes Y)$, 
$(\pm B_n\otimes Y,B_n\otimes X)$, and hence do not commute.
This shows that also non-vanishing Pauli strings of $A_{n+1}$ and $B_{n+1}$ do not commute.
\end{proof}

The proof in Theorem~\ref{theorem:P_Tij_commutes} inspires the following algorithm to decompose $H_M$ into Pauli-strings.
For each item in the list S that the algorithm produces, all Pauli-strings commute.
\begin{algorithm}[ht]
\KwData{Feasible states $B$, transition matrix T (zero diagonal) fulfilling Theorem~\ref{theorem:MixerHamiltonian}
}
\KwResult{Non-vanishing coefficients $c_{i_1,\cdots,i_n}$ of Pauli-strings, Equation~\eqref{eq:Paulidecomp}}
    S=[ ]
    
    \For{{\normalfont \textbf{all}} $1\leq j<k\leq |J|$, s.t.  $(T)_{j,k}\neq 0$}{
        Recursively compute $A_n$ and $B_n$ for $\ket{w}=\ket{x_{J_j}}$ and $\ket{z}=x_{J_k}$ via Equation~\eqref{eq:recursionPauli}
        
        Append $(T)_{j,k} A_n$ to S
    }
    
    check for cancellations and return S
\caption{Decompose $H_M$ given by Equation~\eqref{eq:Hmdefinition} into Pauli-strings directly}\label{alg:PauliDecompRecursive}
\end{algorithm}

We can illustrate the difference between Algorithms~\ref{alg:PauliDecompSecondQ} and~\ref{alg:PauliDecompRecursive} for $B=\{\ket{01},\ket{10}\}$ and $T_{1\leftrightarrow 2}$.
With Algorithm~\ref{alg:PauliDecompSecondQ} we have $P_{1,2}=\frac{1}{4}(X-iY)(X+iY)$ and $P_{2,1}=\frac{1}{4}(X+iY)(X-iY)$ which can be simplified to $S=P_{1,2}+P_{2,1}=\frac{1}{2}(XX+YY)$.
With Algorithm~\ref{alg:PauliDecompRecursive} we have $A_1=X, B_1=Y$ and $S=A_2=\frac{1}{2}(A_1X+B_1Y)=\frac{1}{2}(XX+YY)$ without the need to simplify the expression.

\begin{table}[t]
    \centering
        \begin{tabular}{lrrr}
        \toprule
        &Algorithm~\ref{alg:PauliDecompStandard} &Algorithm~\ref{alg:PauliDecompSecondQ}&Algorithm~\ref{alg:PauliDecompRecursive}\\
        \midrule
        runtime & $\mathcal{O}( 2^{5n} )$ & $\mathcal{O}(n \gamma)$ & $\mathcal{O}(n \gamma)$ \\
        memory & $\mathcal{O}(2^{2n})$ & $\mathcal{O}(n \gamma)$ & $\mathcal{O}(n \gamma)$ \\
        \bottomrule
        \end{tabular}
    \caption{Comparison of the complexity of the two algorithms for $n$ qubits.
    Here, $\gamma$ is the number of nonzero entries of $T$.
    }
    \label{tab:complexity}
\end{table}


As shown above, Trotterizations can also lead to missing transitions.
It is suggested in \cite{hadfield2019quantum} that it is useful to repeat mixers within one mixing step, which corresponds to $r>1$ in Equation~\eqref{eq:mixer_transition}.
However, as we see in Figure~\ref{fig:TOE_trotter}, there can be more efficient ways to get mixers which provide transitions between all pairs of feasible states. 
One way to do so is to construct an exact Trotterization (restricted to the feasible subspace) as described in \cite{Wang2020}.
However, the ultimate goal is \replaced{not}{\textbf{not}} to avoid Trotter\added{ization} errors, but rather to provide transitions between all pairs of feasible states.
We will revisit the topic of Trotterizations in Section~\ref{sec:constrainedmixers} in more detail for each case and show that there are more efficient ways to do so.

\section{Full/Unrestricted mixer}\label{sec:fullmixer}
We start by applying the proposed algorithm to the case without constraints, i.e., for the case $g=0$ in Equation~\eqref{eq:QCBO}, in order to check for consistency and new insight.
We will see that the presented approach is able to reproduce the ``standard" $X$ mixer as one possibility, but provides a more general framework.
For this case $B=\{\ket*{x_j}, j\in J, x_j\in\{0,1\}^n\}, \ J=\{i, 1\leq i \leq 2^n\}$ which means that $\Span{B} = \mathcal{H}$. Furthermore, using Equation~\eqref{eq:H_ETET} we have that $H_{M,B} = T$, since $E$ is the identity.

\subsection{\texorpdfstring{$T_{\Ham(1)}$}{} aka ``standard" full mixer}
The Hamiltonian of the standard full mixer for n qubits can be written as
\begin{equation}
    \begin{split}
        H_M &= \sum_{j\in J} X_j \\
        &=\sum_{j\in J} \left(\ket{0}\bra{0} + \ket{1}\bra{1}\right)^{\otimes (j-1)}
        \otimes
        (\ket{0}\bra{1} + \ket{1}\bra{0})
        \otimes \left(\ket{0}\bra{0} + \ket{1}\bra{1}\right)^{\otimes (n-j)}\\
        &= \sum_{j,k\in J} \left(T_{\Ham(1)}\right)_{j,k}\ket*{x_j}\bra*{x_k}.
    \end{split}
    \label{eq:H_M_full_standard}
\end{equation}
The last identity in Equation~\eqref{eq:H_M_full_standard} shows that $H_M$ is created by the transition matrix given by $T_{\Ham(1)}$.
This assumes that the feasible states in $B$ are ordered from the smallest to the largest integer representation.
\noinitial{
According to Section~\ref{sec:T_Ham1} a feasible initial state is given by uniform superposition $\ket{\phi_0} = \frac{1}{\sqrt{2^n}}\sum_{j\in J} \ket{x_j} = \frac{1}{\sqrt{2^n}} ( 1, 1, \cdots , 1)^T = \ket{+}^{\otimes n}$, which is a well established fact.}

\subsection{All-to-all full mixer}
For $|J|=2^n$ the full mixer $T_A$ can be written as 
$ T_A = \sum_{j=1}^n T_{\Ham(j)}$.
For the case $T_{\Ham(2)}$ the resulting Hamiltonian $H_M$ does not provide transitions between all pairs of feasible states, but we observe that
$H_M = \sum_{j,k\in J} (T_{\Ham(2)})_{j,k}\ket*{x_j}\bra*{x_k} =
\sum_{j_1}^n \sum_{j_2=j_1+1}^n X_{j_1} X_{j_2}$, i.e., $H_M$
consists of all {\scriptsize$\begin{pmatrix}
n \\
2
\end{pmatrix}$} possible pairs of Pauli-strings which contain exactly two $X$s.
For $m\leq n$ this can be further generalized to
\begin{equation}
        H_M = \sum_{j,k\in J} \left(T_{\Ham(m)}\right)_{j,k}\ket*{x_j}\bra*{x_k} =
        \sum_{j_1=1}^n \sum_{j_2=j_1+1}^n \cdots \sum_{j_m=j_{m-1}+1}^n X_{j_1} \cdots X_{j_m},
\end{equation}
which consists of all {\scriptsize$\begin{pmatrix}
n \\
m
\end{pmatrix}$} possible pairs of Pauli-strings with exactly \replaced{$m$}{m} $X$s.
%
%
The resulting mixer Hamiltonian is therefore given by
\begin{equation}
    H_M = \sum_{j_1=1}^n \bigg( X_{j_1} + \sum_{j_2=j_1+1}^n \bigg( X_{j_1} X_{j_2} + 
    \ldots +
    \sum_{j_n=j_{n-1}+1}^n X_{j_1} \cdots X_{j_n}\bigg)\bigg).
\end{equation}
This means that the mixer consists of the standard mixer plus 
$\begin{pmatrix}
n \\
k
\end{pmatrix}
$ applications of Pauli $X$-k strings for $k$ from $2$ to $n$, which is a large overhead \replaced{compared}{with respect} to the standard X-mixer. 

\subsection{(Cyclic) nearest integer full mixer}
The resulting mixer for $T_{\Delta}$/$T_{\Delta,c}$ involve exponentially many Pauli-strings with increasing $n$. The following shows  $T_{\Delta}$ for $1\leq n \leq 4$.
\begin{equation}
\begin{split}
    H_M^{n=1} &= X_1,\\
    H_M^{n=2} &= I_1 \otimes H_M^{n=1} + \frac{1}{2} (X_1 X_2 + Y_1 Y_2),\\
    H_M^{n=3} &= I_1 \otimes H_M^{n=2} + \frac{1}{4} (X_1 X_2 X_3 - X_1 Y_2 Y_3 + Y_1 X_2 Y_3 + Y_1 Y_2 X_3 ),\\
    H_M^{n=4} &= I_1 \otimes H_M^{n=3} + \frac{1}{8} (X_1 X_2 X_3 X_4 - X_1 X_2 Y_3 Y_4 - X_1 Y_2 X_3 Y_4 - X_1 Y_2 Y_3 X_4\\&
    \hphantom{I_1 \otimes H_M^{n=3} + \frac{1}{8} (} \ + Y_1 X_2 X_3 Y_4 + Y_1 X_2 Y_3 X_4 + Y_1 Y_2 X_3 X_4 - Y_1 Y_2 Y_3 Y_4).\\
\end{split}
\end{equation}

\subsection{Comparison and optimality of full mixers}\label{sec:optimalityfull}

\setlength{\tabcolsep}{3pt}
\begin{table}[t]
    \centering
        \begin{tabular}{l rrrrrr r rrrrrr r rrrrrr}
        \toprule
        n& 1 & 2 & 3 & 4 & 5 & 6&\phantom{a} & 1 & 2 & 3 & 4 & 5 & 6&\phantom{a} & 1 & 2 & 3 & 4 & 5 & 6\\
        \midrule
        &\multicolumn{6}{c}{Ham$(T)$} && \multicolumn{6}{c}{\#$U_3$}&& \multicolumn{6}{c}{\#CX $ = \operatorname{Cost}(H_M)$}\\
        $T_{Ham(1)}$    & 2 & 8 & 24 & 64 & 160 & 384  && 1 & 2 & 3 & 4 & 5 & 6 && 0 & 0 & 0 & 0 & 0 & 0 \\
        $T_A$ & 2 & 16 & 96 & 512 & 2560 & 12288  && 1 & 2 & 3 & 4 & 5 & 6 && 0 & 2 & 10 & 34 & 98 & 258 \\
        $T_{\Delta,c}$ & 2 & 12 & 28 & 60 & 124 & 252 && 1 & 1 & 1 & 1 & 1 & 1 && 0 & 2 & 12 & 44 & 132 & 356 \\
        $T_{\Delta}$ & 2 & 8 & 22 & 52 & 114 & 240  && 1 & 1 & 1 & 1 & 1 & 1 && 0 & 4 & 20 & 68 & 196 & 516 \\
        $T_{rand}$ & 2 & 16 & 96 & 512 & 2560 & 12288 && 2 & 4 & 6 & 8 & 10 & 12 && 0 & 10 & 86 & 552 & 3260 & 17650 \\
        \bottomrule
        \end{tabular}
    \caption{Full/unrestricted mixer case for $n$ qubits, i.e., $|B|=2^n$. Comparison of the total Hamming distance of the transition matrix $T$ as well as resulting requirements for implementations in terms of single- and two-qubit gates for different $T$.
    }
    \label{tab:TotalHammingFull}
\end{table}

It would be convenient to have a condition on the transition matrix for the optimality of the resulting mixer.
We define the total Hamming distance of $T$ to be
\begin{equation}
    \text{Ham}(T) = \sum_{\substack{j,k=1\\
    |(T)_{j,k}|>0
    }}^{|J|} d_\text{Hamming}(\text{bin}(i), \text{bin}(j))\replaced{,}{.}
    \label{eq:HamT}
\end{equation}
\added{where $bin(i)$ is the binary representation of an integer i.}
As a first instinct one might suspect that the mixer with minimal Hamming distance, also minimizes the cost.
%
However, this turns out to be false, because of cancellations when more terms in $T$ are non-zero.
Table~\ref{tab:TotalHammingFull} gives a comparison of the total Hamming distance and cost for different full mixers.
The standard full mixer has total Hamming distance $\text{Ham}(T) = (n-1) 2^n$, as there are $2^n$ states \replaced{each with}{with each} $n-1$ states that have Hamming distance 1.
The all to all full mixer has $\text{Ham}(T) = 2^n \sum_{k=1}^n k{\scriptsize\begin{pmatrix}
n \\
k
\end{pmatrix}}$.
\replaced{For the}{The} rest of the transition matrices \replaced{it}{the $Ham(T)$} is not that straightforward to \replaced{derive a general formula for $Ham(T)$}{calculate}, but the table gives an impression.
Table~\ref{tab:TotalHammingFull} shows a dramatic difference between the different mixers with regards to resource requirements.
The standard mixer is the only one which does not require CX gates and is the most efficient to implement.
Furthermore, as the resulting Pauli terms for the full mixers given by $T_{Ham(1)}$ and $T_A$ consist only of $I$ and $X$ and therefore commute they can be implemented without Trotterization. For the mixers given by $\replaced{T_{\Delta,c}}{T_{\Delta}}, T_{\Delta}$, and $T_{rand}$ on the other hand, not all Pauli-strings commute which results in the need for Trotterization.
We continue with the case of constrained mixers.

\section{Constrained mixers}\label{sec:constrainedmixers}
We start by describing what is known as the
``XY"-mixer~\replaced{\cite{fuchs2021efficient,hadfield2019quantum,Wang2020}}{\cite{hadfield2019quantum,Wang2020,fuchs2021efficient}}, before we explore more general cases.
Our framework provides additional insights into this case and inspires further improvement of \replaced{the algorithms above}{above algorithms} with respect to the optimality of the mixers, described in Section~\ref{sec:optimality}, by (possibly) reducing the length of Pauli-strings.
For this case, we will analyze $T_A, T_{\Delta}$, and $T_{\Delta,c}$ only. $T_{Ham(1)}$ only makes sense when $n$ is a power of two, and $T_\text{rand}$ has in general high cost, see Table~\ref{tab:TotalHammingFull}.

\subsection{``One-hot" aka ``XY"-mixer}\label{sec:XYmixer}
We are concerned with the case given by all computational basis states with exactly one ``1", i.e., $B=\{\ket{x}, \ x_j\in\{0,1\}^n, \ s.t. \sum{x_j}=1\}$. These states are sometimes referred to as ``one-hot".
We have that $n=|B|$ is the number of qubits.
After some motivating examples we present the general case for constructing mixers for any $n>2$.
\noinitial{
Feasible initial states depend on $T$ and are given by Corollary~\ref{corollary:feasibleinitialstates}.
For example for $T=T_{\Delta,c}$ or $T=T_A$ a feasible initial state is given by the $W_n$-state
\begin{equation}
    \ket{\phi_0} = \frac{1}{\sqrt{n}}\left(\ket{000\dots001} + \ket{000\dots010} + \dots +\ket{100\dots000}\right).
\end{equation}
An efficient algorithm for this with logarithmic (in $n$) time complexity is presented in~\cite{cruz2019efficient}.
For $T=T_\Delta$ a feasible initial state is given by
\begin{equation}
    \ket{\phi_0}=\frac{1}{\sqrt{|v_n|}}\sin(c)\ket{000\dots001}
    +\sin(2c)\ket{000\dots010} + \dots + \sin(nc)\ket{100\dots000},
\end{equation}
where $c$ and $v_n$ are defined in Equation~\eqref{eq:eigvecT_delta}.
}

\subsubsection{Case \texorpdfstring{$n=2$}{}} The smallest, non-trivial case is given by $B=\{\ket{01},\ket{10}\}$. For any $b\in\R, |b|>0$ the transition matrix $T=\scriptsize\begin{pmatrix}
d& b \\
b& d 
\end{pmatrix}$ fulfills Theorem~\ref{theorem:MixerHamiltonian} and leads to the mixer
$H_M = \frac{b}{2}(XX+YY) + \frac{d}{2}( II - ZZ )$.
Since we want to minimize $\operatorname{Cost}(H_M)$ given in Equation~\eqref{eq:optimality}, we set $d=0$, which results in the  $\operatorname{Cost}(H_M)=4$.
However, by using Corollary~\ref{corollary:properties} there is room for further reducing the cost.
We can add the mixer for $C=\{\ket{00},\ket{11}\}$ since $B\cap C=\varnothing$.
Using the same $T$ (setting $d=0$) gives
\begin{equation}
    H_{M,B} + H_{M,C} = b XX,
\end{equation}
which has $\operatorname{Cost}(H_M)=2$. No Trotterization is needed in this case.

\begin{figure}
    \centering
\begin{tikzpicture}[align=center,node distance=2cm]
\node[circle, draw=black, thick, inner sep=2pt, minimum size=35pt] (1) at (0,0) {$IXX$};
\node[circle, draw=black, thick, inner sep=2pt, minimum size=35pt] (2) [right of=1] {$YYI$};
\node[circle, draw=black, thick, inner sep=2pt, minimum size=35pt] (3) [below right of=2] {$XXI$};
\node[circle, draw=black, thick, inner sep=2pt, minimum size=35pt] (4) [below left of=3] {$YIY$};
\node[circle, draw=black, thick, inner sep=2pt, minimum size=35pt] (5) [left of=4] {$XIX$};
\node[circle, draw=black, thick, inner sep=2pt, minimum size=35pt] (6) [above left of=5] {$IYY$};

\draw [thick,-,green] (2) -- (3);
\draw [thick,-,green] (4) -- (5);
\draw [thick,-,green] (1) -- (6);

\draw [thick,-,red] (2) -- (4) -- (6) -- (2);
\draw [thick,-,blue] (1) -- (3) -- (5) -- (1);

\node [single arrow, draw=black,fill=black!25] at (4.8,-1.4) {\phantom{aaa}};
\node [rotate=180,single arrow, draw=black,fill=black!25] at (-2.8,-1.4) {\phantom{aaa}};

\node[circle, draw=green, thick, inner sep=2pt, minimum size=35pt] (A) at (-6,-.5) {$XXI$\\$YYI$};
\node[circle, draw=green, thick, inner sep=2pt, minimum size=35pt] (B) at (-4.5,-1.5) {$XIX$\\$YIY$};
\node[circle, draw=green, thick, inner sep=2pt, minimum size=35pt] (C) at (-6,-2.5) {$IXX$\\$IYY$};

\node[circle, draw=blue, thick, inner sep=2pt, minimum size=35pt] (D) at (6.5,-.5) {$IXX$\\$XIX$\\$XXI$};
\node[circle, draw=red, thick, inner sep=2pt, minimum size=35pt] (E) at (6.5,-2.5) {$IYY$\\$YIY$\\$YYI$};
\end{tikzpicture}
    \caption{In the commutation graph (middle) of the terms of the mixer given in Equation~\eqref{eq:XYn3} an edge occurs if the terms commute.
    From this we can group terms into three (nodes connected by green edge) or two (nodes connected by red/blue edges) sets. Only the left/green grouping preserves the feasible subspace.}
    \label{fig:XYn3}
\end{figure}
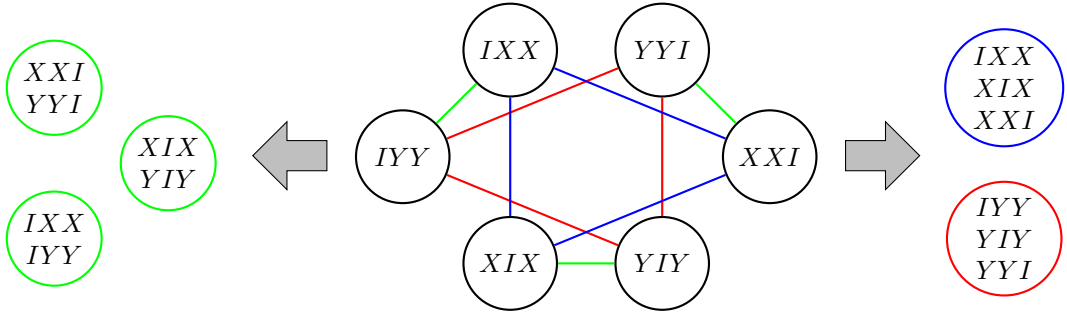
\subsubsection{Case \texorpdfstring{$n=3$}{}} We continue with $B=\{\ket{001}, \ket{010}, \ket{100}\}$. For the transition matrix $T= aT_{1\leftrightarrow 2}+bT_{2\leftrightarrow 3}+cT_{1\leftrightarrow 3}, a,b,c\in\R$
this results in the mixer
\begin{equation}
    H_{M,B} = \frac{a}{4} (XX+YY)(I+Z) + \frac{b}{4} (I+Z)(XX+YY) + \frac{c}{4} (X(I+Z)X+Y(I+Z)Y),
\end{equation}
with associated cost $\operatorname{Cost}(H_M)=24+12c$ for $a=b=1,c\in\{0,1\}$.
In this case Corollary~\ref{corollary:properties} allows us to add the mixer for $C=\{\ket{110}, \ket{101}, \ket{011}\}$ since $B\cap C=\varnothing$.
The mixer
\begin{equation}
   H_M = H_{M,B} + H_{M,C} = \frac{a}{2} (XXI+YYI) + \frac{b}{2} (XIX+YIY)+ \frac{c}{2} (IXX+IYY),
   \label{eq:XYn3}
\end{equation}
has cost  $\operatorname{Cost}(H_M)=8+4c$ for $a=b=1,c\in\{0,1\}$.
However, this mixer can not be realized, since not all terms of $H_M$ commute. Figure~\ref{fig:XYn3} shows two ways to put the graph into commuting Pauli-terms, with only one way to preserve the feasible subspace, as discussed in Section~\ref{sec:Trotterizations}.
\replaced{
For the Trotterization according to $T=T_{1\leftrightarrow 2}^{k_1}+T_{2\leftrightarrow 3}^{k_2}$ we have that $(T)_{3,1}=(T)_{1,3}=0, \ \forall k_1,k_2\in\N$. To fulfill Theorem~\ref{theorem:TrotterizedMixerHamiltonian} we need to include the term $T_{1\leftrightarrow 3}$ as well. The Trotterized mixer with minimal cost is therefore given by $T= T_{1\leftrightarrow 2}+T_{2\leftrightarrow 3}+T_{1\leftrightarrow 3}$.}{A Trotterized mixer with the lowest}

\subsubsection{The general case \texorpdfstring{$n>2$}{}}
We start with the observation that
for any symmetric $T\in \R^{n\times n}$ with zero diagonal we have
\begin{equation}
    H_{M,B} = \sum_{j=1}^{n} \sum_{k=j+1}^{n} (T)_{j,k} \widehat P^{j,k}, \quad
    (\widehat P^{j,k})_l = \frac{1}{2^n}\begin{cases}
    (X+Y) ,& \text{if } l\in\{j,k\}, \\
    (I+Z), &\text{if } x_l=z_l=0
    \end{cases}
\end{equation}
The cost for implementing on\added{e} of the entries, i.e., $e^{-i\beta \hat P^{j,k}}$ is given by the recursive formula
\begin{equation}
        \operatorname{Cost}(\widehat  P^{j,k}) = \sum_{l=2}^n 2 (l-1) f_n^l, \ n>2, \quad
        f_n^l = f_{n-1}^l + f_{n-1}^{l-1}, \quad
        f_2^l = \begin{cases}
            2, & \text{\hspace{-4pt}if } l=2,\\
            0, & \text{\hspace{-4pt}else},
        \end{cases}
\end{equation}
where $f_n^l$ is Pascal's triangle starting with 2 instead of 1.
Example\replaced{s of}{for} the resulting costs for different transition matrices can be seen in Table~\ref{tab:XY}.



\setlength{\tabcolsep}{4pt}
\begin{table}[t]
    \centering
        \begin{tabular}{l rrrrrrrrr}
        \toprule
        n&  3 & 4 & 5 & 6 & 7 & 8 & 9 & 10 & 15\\
        \midrule
        \multicolumn{2}{l}{$\mathbf{H_{M,B}}$}\\
        $T_{\Delta}$    &  12$\cdot$2& 32$\cdot$3& 80$\cdot$4\phantom{4}& 192$\cdot$5\phantom{4}& 448$\cdot$6\phantom{4}& 1024$\cdot$7\phantom{4}& 2304$\cdot$8\phantom{4}& 5120$\cdot$9\phantom{4} & 245760$\cdot$14\phantom{4}\\
        $T_{\Delta,c}$ & 12$\cdot$3& 32$\cdot$4& 80$\cdot$5\phantom{4}& 192$\cdot$6\phantom{4}& 448$\cdot$7\phantom{4}& 1024$\cdot$8\phantom{4}& 2304$\cdot$9\phantom{4}& 5120$\cdot$10 & 245760$\cdot$15\phantom{4}\\
        $T_A$           & 12$\cdot$3& 32$\cdot$6& 80$\cdot$10& 192$\cdot$15& 448$\cdot$21& 1024$\cdot$28& 2304$\cdot$36& 5120$\cdot$45 & 245760$\cdot$105\\
        \multicolumn{2}{l}{$\mathbf{H_{M,B}+\sum_i H_{M,C_i}}$}\\
        $T_{\Delta}$     & 4$\cdot$2& 4$\cdot$3& 4$\cdot$4\phantom{4}& 4$\cdot$5\phantom{4}& 4$\cdot$6\phantom{4}& 4$\cdot$7\phantom{4}& 4$\cdot$8\phantom{4}& 4$\cdot$9\phantom{4}& 4$\cdot$14\phantom{4}\\
        $T_{\Delta,c}$ & 4$\cdot$3& 4$\cdot$4& 4$\cdot$5\phantom{4}& 4$\cdot$6\phantom{4}& 4$\cdot$7\phantom{4}& 4$\cdot$8\phantom{4}& 4$\cdot$9\phantom{4}& 4$\cdot$10& 4$\cdot$15\phantom{4}\\
        $T_A$           & 4$\cdot$3& 4$\cdot$6& 4$\cdot$10& 4$\cdot$15& 4$\cdot$21& 4$\cdot$28& 4$\cdot$36& 4$\cdot$45& 4$\cdot$105\\
        \bottomrule
        \end{tabular}
    \caption{
    Comparison of the cost \#CX $ = \operatorname{Cost}(H_M)$ of mixers constrained to ``one-hot" states.
    The Trotterized versions we define $T_1=T_{O(1),c}$ and $T_2=T_{E(1)}$.
    All Hamiltonians need to be Trotterized.
    }
    \label{tab:XY}
\end{table}

The cost of the mixers can be considerably reduced by adding mixers generalized from case $n=3$.
If the entries $(T)_{i\leftrightarrow j}$ of $T$ are non-zero, we can add mixers for each of the $2^{n-2}$ pairs of states $x\in\{0,1\}^n$ that fulfill that $(x_i=0 \wedge x_j=1) \vee (x_i=1 \wedge x_j=0)$.
We can enumerate them with $0\leq l \leq 2^{n-2}-1$ by $\widetilde B_{i,j}^l=\{ \ket{x}, x\in\{0,1\}^n, \text{ s.t. } x_{-i,-j} = bin(l)\}$, where $x_{-i,-j}$ removes the indices i and j of x.
We have that $B\cap \widetilde B^l_{i,j}=\varnothing$.
%
We observe that for $n\geq 2$ let $\ket{x}, \ket{z}$ with $Ham(x,z)=2$, i.e., the strings $x,z$ differ at exactly two positions we have that
\begin{equation}
    \ket{x}\bra{z}+\ket{z}\bra{x} = \frac{1}{2^{n-1}}\begin{cases}
    (X+Y) ,& \text{if } x_l\neq z_l, \\
    (I+Z), &\text{if } x_l=z_l=0,\\
    (I-Z), &\text{if } x_l=z_l=1.
\end{cases}
\end{equation}
%
Adding these mixers for each nonzero entries $T_{j\leftrightarrow k}$ of $T$ has\deleted{t} the effect \replaced{of summing}{to sum} over all possible combinations of $(I\pm Z)^{\otimes{2^{n-2}}}$ which is equal to the identity. Therefore, we get the mixer
\begin{equation}\label{eq:XYmixer}
H_{M,B} + \sum_{i,j\in J}\sum_{l=0}^{2^{n-2}} H_{M,B^l_{i,j}} = \sum_{j=1}^{n} \sum_{k=j+1}^{n} (T)_{j,k} P^{j,k}, \quad P^{j,k}=X_iX_j+Y_iY_j,
\end{equation}
which reduces the cost of one term to $\operatorname{Cost}(P^{j,k})= 4$.




\subsubsection{Trotterizations}
Not all Pauli-strings of the mixer in Equation~\eqref{eq:XYmixer} commute. This necessitates a suitable and efficient Trotterization.
We will use Theorem~\ref{theorem:TrotterizedMixerHamiltonian} and Theorem~\ref{theorem:P_Tij_commutes} to identify valid Trotterized mixers.
As pointed out in~\cite{Wang2020} when \textit{n is a power of two} one can realize a Trotterization which is exact in the feasible subspace $B$.
Termed \textit{simultaneous complete-graph} mixer, this involves all possible pairs $(i,j)$ corresponding to a certain Trotterization of mixer for $T_A$. We will see that there are more efficient mixers that provide transitions between all pairs of feasible states.

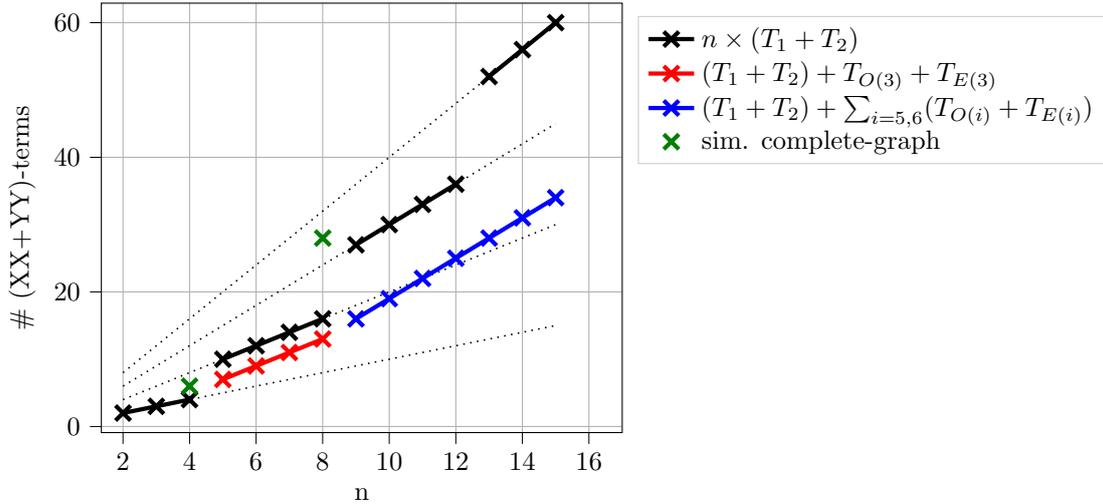
\begin{figure}
    \centering
\begin{tikzpicture}

\definecolor{darkgray176}{RGB}{176,176,176}
\definecolor{green01270}{RGB}{0,127,0}
\definecolor{lightgray204}{RGB}{204,204,204}

\begin{axis}[
legend cell align={left},
legend style={
  fill opacity=0.8,
  draw opacity=1,
  text opacity=1,
  at={(1.03,0.97)},
  anchor=north west,
  draw=lightgray204
},
tick align=outside,
tick pos=left,
x grid style={darkgray176},
xlabel={n},
xmajorgrids,
xmin=1.35, xmax=17,
xtick style={color=black},
y grid style={darkgray176},
ylabel={\# (XX+YY)-terms},
ymajorgrids,
ymin=-0.9, ymax=62.9,
ytick style={color=black}
]
\addplot [semithick, black, dotted, forget plot]
table {%
2 2
3 3
4 4
5 5
6 6
7 7
8 8
9 9
10 10
11 11
12 12
13 13
14 14
15 15
};
\addplot [ultra thick, black, mark=x, mark size=4, mark options={solid}, forget plot]
table {%
2 2
3 3
4 4
};
\addplot [semithick, black, dotted, forget plot]
table {%
2 4
3 6
4 8
5 10
6 12
7 14
8 16
9 18
10 20
11 22
12 24
13 26
14 28
15 30
};
\addplot [ultra thick, black, mark=x, mark size=4, mark options={solid}, forget plot]
table {%
5 10
6 12
7 14
8 16
};
\addplot [semithick, black, dotted, forget plot]
table {%
2 6
3 9
4 12
5 15
6 18
7 21
8 24
9 27
10 30
11 33
12 36
13 39
14 42
15 45
};
\addplot [ultra thick, black, mark=x, mark size=4, mark options={solid}, forget plot]
table {%
9 27
10 30
11 33
12 36
};
\addplot [semithick, black, dotted, forget plot]
table {%
2 8
3 12
4 16
5 20
6 24
7 28
8 32
9 36
10 40
11 44
12 48
13 52
14 56
15 60
};
\addplot [ultra thick, black, mark=x, mark size=4, mark options={solid}]
table {%
13 52
14 56
15 60
};
\addlegendentry{$n\times(T_1+T_2)$}
\addplot [ultra thick, red, mark=x, mark size=4, mark options={solid}]
table {%
5 7
6 9
7 11
8 13
};
\addlegendentry{$(T_1+T_2)+T_{O(3)}+T_{E(3)}$}
\addplot [ultra thick, blue, mark=x, mark size=4, mark options={solid}]
table {%
9 16
10 19
11 22
12 25
13 28
14 31
15 34
};
\addlegendentry{$(T_1+T_2)+\sum_{i=5,6}(T_{O(i)}+T_{E(i)})$}
\addplot [ultra thick, green01270, mark=x, mark size=4, mark options={solid}, only marks]
table {%
4 6
8 28
};
\addlegendentry{sim. complete-graph}
\end{axis}

\end{tikzpicture}
    \caption{Comparison of different Trotterization mixers restricted to ``one-hot" states.
    All markers represent cases when the resulting mixer provides transition\added{s} for all pairs of feasible states, see also Figure~\ref{fig:TOE_trotter}.
    All versions can be implemented in linear depth.
    The most efficient Trotterizations are achieved by using sub-diagonal entries.
    The cost equals 4 times \# (XX+YY)-terms.
    }
    \label{fig:XY-trotter}
\end{figure}

Another possibility is to Trotterize $T_{\Delta,c}$ or $T_\Delta$ according to odd and even entries as described in Section~\ref{sec:Toddeven}. This is what is termed \added{a} \textit{parity-partitioned} mixer in~\cite{Wang2020}. However, fewer and fewer feasible states can be reached as $n$ increases, as we have seen in Figure~\ref{fig:TOE_trotter}.
Repeated applications ($r>0$ in Equation~\eqref{eq:mixer_transition}) are necessary and $r$ increases with increasing $n$.
Figure~\ref{fig:XY-trotter} shows a comparison of different Trotterizations.
As the cost of the mixer is dictated by the number of non-zero entries of the transition matrix, it is more efficient to \textit{add} mixers for off-diagonals according to
$\sum_{i\in I} (T_{O(i)}\!  +\!  T_{E(i)})$ for some suitable index set $I$.

\subsection{General cases}\label{sec:generalcases}
In this section we analyze some specific cases that go beyond unrestricted mixers and mixers restricted to one-hot states.

\setlength{\tabcolsep}{3pt}
\begin{table}[t]
    \centering
    
    \begin{tabular}{lrrr}
        \toprule
        $C=$ & $T_{1 \leftrightarrow 2}$ & $T_{{2 \leftrightarrow 3}}$ & $T_{{3 \leftrightarrow 1}}$ \\
        \midrule
        $\{\}$ & 12 & 8 & 16\\
        $\{\ket{000},\ket{001}\}$ & 20 & \cellcolor{black!25}2 & 24\\
        $\{\ket{000},\ket{101}\}$ & 24 & 20 & 28\\
        $\{\ket{000},\ket{110}\}$ & \cellcolor{black!25}6 & 20 & 28\\
        $\{\ket{000},\ket{111}\}$ & 28 & 24 & \cellcolor{black!25}8\\
        $\{\ket{001},\ket{101}\}$ & 20 & 16 & 24\\
        $\{\ket{001},\ket{110}\}$ & 28 & 24 & \cellcolor{black!25}8\\
        $\{\ket{001},\ket{111}\}$ & \cellcolor{black!25}6 & 20 & 28\\
        $\{\ket{101},\ket{110}\}$ & 24 & 20 & 28\\
        $\{\ket{101},\ket{111}\}$ & 20 & 16 & 24\\
        $\{\ket{110},\ket{111}\}$ & 20 & \cellcolor{black!25}2 & 24\\
        \bottomrule
    \end{tabular}
    \caption{
    Comparison of the $\operatorname{Cost}(H_{M,B}+H_{M,C})$ for different added mixers $C$ for the case $B=\{\ket{100}, \ket{010}, \ket{011}\}$.
    All 10 possible pairs are shown.
    We see that the cost can be both reduced and increased.
    }
    \label{tab:case1_addmixer}
    \end{table}

\subsubsection{Example 1}
We start by looking at the case $B=\{\ket{100}, \ket{010}, \ket{011}\}$.
Using $T_\Delta=c_{1,2} T_{1\leftrightarrow2}+c_{2,3} T_{2\leftrightarrow3}$ and $T_{\Delta,c}=T_\Delta+c_{3,1} T_{3\leftrightarrow1}$, this results in the mixer
\begin{equation}
\begin{split}
    H_{M,B}=& \phantom{ + } c_{1,2}\frac{1}{4}(XX+YY)(I+Z)\\
    &+ c_{2,3}\frac{1}{4}(I+Z)(I-Z)X\\
    & + c_{3,1}\frac{1}{4}(XXX + YXY + YYX - XYY)
\end{split}
\end{equation}
with $\operatorname{Cost}(H_M) = 12c_{1,2}+8c_{2,3}+16c_{3,1}$.
Here, $(c_{1,2},c_{2,3},c_{3,1})=(1,1,0)$ corresponds to $T_\Delta$
and $(1,1,1)$ to $T_{\Delta,c}$.
There is a lot of freedom adding mixers, which is summarized in 
Table~\ref{tab:case1_addmixer}.
Adding more terms only increases the cost for this case.
%
Overall, the most efficient mixers for $B$ \added{is} given by
\begin{equation}
    H_{M} = \frac{c_{1,2}}{2}
    \begin{bmatrix*}[l]
    XXI + XXZ \text{ or} \\
    XXI + YYZ
    \end{bmatrix*}
    + \frac{c_{2,3}}{2}
    \begin{bmatrix*}[l]
    (I+Z)IX \text{ or} \\
    I(I-Z)X
    \end{bmatrix*}
    + \frac{c_{3,1}}{2}
    \begin{bmatrix*}[l]
    XXX - XYY \text{ or} \\
     XXX + YXY
    \end{bmatrix*},
\end{equation}
with associated cost 
$\operatorname{Cost}(H_M) = 6c_{1,2}+2c_{2,3}+8c_{3,1}$.
A valid Trotterization is given through splitting according to $T_{i\leftrightarrow j}$.





\subsubsection{Example 2}

\setlength{\tabcolsep}{0pt}
   \begin{table}[t]
    \centering
    \begin{tabular}{lrrrrrrrrrrrrrrr}
        \toprule
$C=$ & $T_{1 \leftrightarrow 2}$ & $T_{1 \leftrightarrow 3}$ & $T_{{1 \leftrightarrow 4}}$ & $T_{{1 \leftrightarrow 5}}$& $T_{{1 \leftrightarrow 6}}$&
       $T_{{2 \leftrightarrow 3}}$ & $T_{{2 \leftrightarrow 4}}$ & $T_{{2 \leftrightarrow 5}}$& $T_{{2 \leftrightarrow 6}}$&
       $T_{{3 \leftrightarrow 4}}$ & $T_{{3 \leftrightarrow 5}}$& $T_{{3 \leftrightarrow 6}}$&
       $T_{{4 \leftrightarrow 5}}$ & $T_{{4 \leftrightarrow 6}}$&
       $T_{{5 \leftrightarrow 6}}$
       \\
\midrule
$\scriptstyle\{\}$ & 96 & 64 & 112 & 80 & 80 & 112 & 96 & 64 & 64 & 96 & 96 & 96 & 112 & 112 & 80\\
$\scriptstyle\{\ket{00010},\ket{00011}\}$ & 160 & \cellcolor{black!25}24 & 176 & 144 & 144 & 176 & 160 & 128 & 128 & 160 & 160 & 160 & 176 & 176 & 144\\
$\scriptstyle\{\ket{00010},\ket{01101}\}$ & 208 & 176 & \cellcolor{black!25}48 & 192 & 192 & 224 & 208 & 176 & 176 & 208 & 208 & 208 & 224 & 224 & 192\\
$\scriptstyle\{\ket{00010},\ket{10001}\}$ & 192 & 160 & 208 & 176 & 176 & 208 & \cellcolor{black!25}48 & 160 & 160 & 192 & 192 & 192 & 208 & 208 & 176\\
$\scriptstyle\{\ket{00010},\ket{10101}\}$ & 208 & 176 & 224 & 192 & 192 & 224 & 208 & 176 & 176 & 208 & 208 & 208 & 224 & \cellcolor{black!25}48 & 192\\
$\scriptstyle\{\ket{00010},\ket{11001}\}$ & 208 & 176 & 224 & 192 & 192 & 224 & 208 & 176 & 176 & 208 & 208 & 208 & \cellcolor{black!25}48 & 224 & 192\\
$\scriptstyle\{\ket{00011},\ket{01101}\}$ & 192 & 160 & 208 & 176 & 176 & 208 & 192 & 160 & 160 & \cellcolor{black!25}40 & 192 & 192 & 208 & 208 & 176\\
$\scriptstyle\{\ket{00011},\ket{10000}\}$ & 192 & 160 & 208 & 176 & 176 & 208 & \cellcolor{black!25}48 & 160 & 160 & 192 & 192 & 192 & 208 & 208 & 176\\
$\scriptstyle\{\ket{00100},\ket{01000}\}$ & 176 & 144 & 192 & 160 & 160 & 192 & 176 & 144 & 144 & 176 & 176 & 176 & 192 & 192 & \cellcolor{black!25}32\\
$\scriptstyle\{\ket{00100},\ket{01100}\}$ & 160 & 128 & 176 & 144 & 144 & 176 & 160 & \cellcolor{black!25}24 & 128 & 160 & 160 & 160 & 176 & 176 & 144\\
$\scriptstyle\{\ket{00100},\ket{10000}\}$ & 176 & 144 & 192 & \cellcolor{black!25}32 & 160 & 192 & 176 & 144 & 144 & 176 & 176 & 176 & 192 & 192 & 160\\
$\scriptstyle\{\ket{00100},\ket{10001}\}$ & 192 & 160 & 208 & 176 & 176 & 208 & 192 & 160 & 160 & 192 & \cellcolor{black!25}40 & 192 & 208 & 208 & 176\\
$\scriptstyle\{\ket{00100},\ket{10111}\}$ & 192 & 160 & 208 & 176 & 176 & 208 & \cellcolor{black!25}48 & 160 & 160 & 192 & 192 & 192 & 208 & 208 & 176\\
$\scriptstyle\{\ket{00101},\ket{10110}\}$ & 192 & 160 & 208 & 176 & 176 & 208 & \cellcolor{black!25}48 & 160 & 160 & 192 & 192 & 192 & 208 & 208 & 176\\
$\scriptstyle\{\ket{00111},\ket{01011}\}$ & 176 & 144 & 192 & 160 & 160 & 192 & 176 & 144 & 144 & 176 & 176 & 176 & 192 & 192 & \cellcolor{black!25}32\\
$\scriptstyle\{\ket{00111},\ket{01111}\}$ & 160 & 128 & 176 & 144 & 144 & 176 & 160 & \cellcolor{black!25}24 & 128 & 160 & 160 & 160 & 176 & 176 & 144\\
$\scriptstyle\{\ket{00111},\ket{10100}\}$ & 192 & 160 & 208 & 176 & 176 & 208 & \cellcolor{black!25}48 & 160 & 160 & 192 & 192 & 192 & 208 & 208 & 176\\
$\scriptstyle\{\ket{01000},\ket{01100}\}$ & 160 & 128 & 176 & 144 & 144 & 176 & 160 & 128 & \cellcolor{black!25}24 & 160 & 160 & 160 & 176 & 176 & 144\\
$\scriptstyle\{\ket{01000},\ket{10000}\}$ & 176 & 144 & 192 & 160 & \cellcolor{black!25}32 & 192 & 176 & 144 & 144 & 176 & 176 & 176 & 192 & 192 & 160\\
$\scriptstyle\{\ket{01000},\ket{10001}\}$ & 192 & 160 & 208 & 176 & 176 & 208 & 192 & 160 & 160 & 192 & 192 & \cellcolor{black!25}40 & 208 & 208 & 176\\
$\scriptstyle\{\ket{01000},\ket{11011}\}$ & 192 & 160 & 208 & 176 & 176 & 208 & \cellcolor{black!25}48 & 160 & 160 & 192 & 192 & 192 & 208 & 208 & 176\\
$\scriptstyle\{\ket{01001},\ket{11010}\}$ & 192 & 160 & 208 & 176 & 176 & 208 & \cellcolor{black!25}48 & 160 & 160 & 192 & 192 & 192 & 208 & 208 & 176\\
$\scriptstyle\{\ket{01011},\ket{01111}\}$ & 160 & 128 & 176 & 144 & 144 & 176 & 160 & 128 & \cellcolor{black!25}24 & 160 & 160 & 160 & 176 & 176 & 144\\
$\scriptstyle\{\ket{01011},\ket{11000}\}$ & 192 & 160 & 208 & 176 & 176 & 208 & \cellcolor{black!25}48 & 160 & 160 & 192 & 192 & 192 & 208 & 208 & 176\\
$\scriptstyle\{\ket{01100},\ket{10000}\}$ & \cellcolor{black!25}40 & 160 & 208 & 176 & 176 & 208 & 192 & 160 & 160 & 192 & 192 & 192 & 208 & 208 & 176\\
$\scriptstyle\{\ket{01100},\ket{10001}\}$ & 208 & 176 & 224 & 192 & 192 & \cellcolor{black!25}48 & 208 & 176 & 176 & 208 & 208 & 208 & 224 & 224 & 192\\
$\scriptstyle\{\ket{01100},\ket{11111}\}$ & 192 & 160 & 208 & 176 & 176 & 208 & \cellcolor{black!25}48 & 160 & 160 & 192 & 192 & 192 & 208 & 208 & 176\\
$\scriptstyle\{\ket{01101},\ket{11110}\}$ & 192 & 160 & 208 & 176 & 176 & 208 & \cellcolor{black!25}48 & 160 & 160 & 192 & 192 & 192 & 208 & 208 & 176\\
$\scriptstyle\{\ket{01111},\ket{11100}\}$ & 192 & 160 & 208 & 176 & 176 & 208 & \cellcolor{black!25}48 & 160 & 160 & 192 & 192 & 192 & 208 & 208 & 176\\
$\scriptstyle\{\ket{10000},\ket{10001}\}$ & 160 & \cellcolor{black!25}24 & 176 & 144 & 144 & 176 & 160 & 128 & 128 & 160 & 160 & 160 & 176 & 176 & 144\\
$\scriptstyle\{\ket{10110},\ket{10111}\}$ & 160 & \cellcolor{black!25}24 & 176 & 144 & 144 & 176 & 160 & 128 & 128 & 160 & 160 & 160 & 176 & 176 & 144\\
$\scriptstyle\{\ket{10110},\ket{11010}\}$ & 176 & 144 & 192 & 160 & 160 & 192 & 176 & 144 & 144 & 176 & 176 & 176 & 192 & 192 & \cellcolor{black!25}32\\
$\scriptstyle\{\ket{10110},\ket{11110}\}$ & 160 & 128 & 176 & 144 & 144 & 176 & 160 & \cellcolor{black!25}24 & 128 & 160 & 160 & 160 & 176 & 176 & 144\\
$\scriptstyle\{\ket{11010},\ket{11011}\}$ & 160 & \cellcolor{black!25}24 & 176 & 144 & 144 & 176 & 160 & 128 & 128 & 160 & 160 & 160 & 176 & 176 & 144\\
$\scriptstyle\{\ket{11010},\ket{11110}\}$ & 160 & 128 & 176 & 144 & 144 & 176 & 160 & 128 & \cellcolor{black!25}24 & 160 & 160 & 160 & 176 & 176 & 144\\

        $\phantom{\{\ket{0000}\}}\cdots$\\
$\scriptstyle\{\ket{00000},\ket{11111}\}$ & 224 & 192 & 240 & 208 & 208 & 240 & 224 & 192 & 192 & 224 & 224 & 224 & 240 & 240 & 208\\
        $\phantom{\{\ket{0000}\}}\cdots$\\
        \bottomrule
    \end{tabular}
    \caption{
    Comparison of the $\operatorname{Cost}(H_{M,B}+H_{M,C})$ for different added mixers $C$
    for the case $B=\{\ket{10010},\allowbreak \ket{01110}, \allowdisplaybreaks \ket{10011},\allowbreak \ket{11101},\allowbreak \ket{00110},\allowbreak \ket{01010}\}$.
    There are 325 possible pairs in total. We see that the cost can be both reduced and increased.
    }
    \label{tab:case2_addmixer}
    \end{table}

Finally, we investigate the case $B=\{\ket{10010},\ket{01110},\ket{10011},\ket{11101},\ket{00110},\ket{01010}\}$, which restricts to 6 of the total $2^5=32$ computational basis states for 5 qubits.
It is not \replaced{clear a priori}{a priorily clear} if for any (distinct) pair $T_{i_1,j_1}$ and $T_{i_2,j_2}$ all pairs of non-vanishing Pauli strings commute. In order to fulfill Equation~\eqref{eq:mixer_transition} for $r=1$, this means that one needs to Trotterize according to all pairs of $T_A$ as shown in Table~\ref{tab:case2_addmixer}.
The resulting cost for this Trotterized mixer is $\operatorname{Cost}(H_M)=1360$.
Since $\mathcal{H}\cap B$ is spanned by $k=2^n-|B|=26$ computational basis states,
there are
{\scriptsize$\begin{pmatrix}
k \\
2
\end{pmatrix}=325$} different pairs to add to each $T_{i\leftrightarrow j}$.
As Table~\ref{tab:case2_addmixer} shows this can reduce the cost of the resulting mixer to 
$\operatorname{Cost}(H_M)=568$.
Of course, there is the possibility to reduce the cost even further by add\added{ing} more mixers for states in the kernel of $H_{M,B}$.
However, this quickly becomes computationally very demanding, when all possibilities are \replaced{considered in a brute-force fashion}{brute-forced}.

\section{Availability of Data and Code}
All data and the python/jupyter notebook source code for reproducing the results obtained in this article are available at  \url{https://github.com/OpenQuantumComputing}.

\section{Conclusion and Outlook}
While designing mixers with the presented framework is more or less straight forward, designing \textit{efficient} mixers turns out to be a difficult task\deleted{s}.
An additional difficulty arises due to the need for Trotterization.
Somewhat counter-intuitively, the more restricted the mixer, i.e., the smaller the subspace, the more design freedom one has to increase efficiency.
More structure/symmetry of the restricted subspace seems to allow for a lower cost of the resulting mixer. 
For the case of ``one-hot" states, we provide a deeper understanding of the requirements for Trotterizations. Compared to state of the art in \added{the} literature, this leads to a considerable reduction of the cost of the mixer, as defined in Equation~\eqref{eq:optimality}.
The introduced framework reveals a rigorous mathematical analysis of the underlying structure of mixer Hamiltonians and deepens the understanding of those.
We believe the framework can serve as the backbone for further development of efficient mixers.

When adding mixers, in general the kernel of $H_{M,B}$ is spanned by
$k=2^n-|B|$ computational basis states. Therefore, one can add 
\begin{equation}
    \sum_{i=2}^k \begin{pmatrix}
k \\
2
\end{pmatrix}
\end{equation}
different mixers for each non-zero entry $T_{i\leftrightarrow j}$ of $T$.
Out of all these, one wants to find the combination leading to the lowest overall cost.
Clearly, brute-forc\replaced{e}{ing the} optimization is computationally not tractable, even for \added{a} moderate \deleted{sized} number of qubits $n$ when $|B|\ll2^n$.
Further research should aim to carefully analyze \deleted{of} the structure of the basis states in $B$ in order to develop efficient (heuristic) algorithms to find low-cost mixers through adding mixers in the kernel of $H_{M,B}$.


\printbibliography
\end{document}